\documentclass[a4paper,UKenglish,thm-restate]{lipics-v2021}
\usepackage[utf8]{inputenc}
\usepackage{amsmath}
\usepackage{amsfonts}
\usepackage{amsthm}
\usepackage{graphicx}
\usepackage{wrapfig}
\usepackage{makecell}
\usepackage{xspace}
\usepackage{caption}
\usepackage{imakeidx}

\usepackage[dvipsnames]{xcolor}

\usepackage{todonotes}
\usepackage[shortlabels]{enumitem}

\hideLIPIcs
\nolinenumbers

\newcommand{\mypar}[1]{\medskip\noindent{\sffamily\bfseries\boldmath#1}}

\graphicspath{{./Figures/}}

\newcommand{\frechet}{Fr\'{e}chet\xspace}

\DeclareMathOperator{\opt}{{\rm opt}}

\title{Barking dogs: A \frechet distance variant for detour detection}
\titlerunning{Barking dogs: A \frechet distance variant for detour detection}

\author{Ivor van der Hoog}{Technical University of Denmark, Denmark}{idjva@dtu.dk}{https://orcid.org/0009-0006-2624-0231}{This project has received funding from the European Union's Horizon 2020 research and innovation programme under the Marie Sk\l{}odowska-Curie grant agreement No 899987.}
\author{Fabian Klute}{Universitat Politècnica de Catalunya, Spain}{fabian.klute@upc.edu}{https://orcid.org/0000-0002-7791-3604}{F.K. is supported by a ``María Zambrano grant for attracting international talent''.}
\author{Irene Parada}{Universitat Politècnica de Catalunya, Spain}{irene.parada@upc.edu}{https://orcid.org/0000-0003-3147-0083}{I.P. is a Serra Húnter fellow.}
\author{Patrick Schnider}{Department of Computer Science, ETH Z\"{u}rich, Switzerland}{patrick.schnider@inf.ethz.ch}{https://orcid.org/0000-0002-2172-9285}{}

\Copyright{Ivor van der Hoog, Fabian Klute, Irene Parada, and Patrick Schnider} 
\authorrunning{I. van der Hoog, F. Klute, I. Parada, and P. Schnider} 

\ccsdesc[500]{ Theory of computation ~ Design and analysis of algorithms}

\keywords{
Outlier detection,
similarity measures,
Dynamic Time Warping,
Fr\'echet distance,
algorithms,
computational geometry
}

\category{}

\relatedversion{}

\supplement{}

\acknowledgements{}

\funding{}

\newtheorem{invariant}{Invariant}

\begin{document}
\maketitle

\begin{abstract}
Imagine you are a dog behind a fence $Q$ and a hiker is passing by at constant speed along the hiking path $P$.
In order to fulfil your duties as a watchdog, you desire to bark as long as possible at the human.
However, your barks can only be heard in a fixed radius $\rho$ and, as a dog, you have bounded speed $s$. 
Can you optimize your route along the fence $Q$ in order to maximize the barking time with radius $\rho$, 
assuming you can run backwards and forward at speed at most $s$?

We define the barking distance from a polyline $P$ on $n$ vertices to a polyline $Q$ on $m$ vertices as the time
that the hiker stays in your barking radius if you run optimally along $Q$.
This asymmetric similarity measure between two curves can be used to detect outliers in $Q$ compared to $P$ that
other established measures like the Fr\'echet distance and Dynamic Time Warping fail to capture at times.
We consider this measure in three different settings. 
In the discrete setting, the traversals of $P$ and $Q$ are both discrete. 
For this case we show that the barking distance from $P$ to $Q$ can be computed in $O(nm\log s)$ time.
In the semi-discrete setting, the traversal of $Q$ is continuous while the one of $P$ is again discrete.
Here, we show how to compute the barking distance in time $O(nm\log (nm))$.
Finally, in the continuous setting in which both traversals are continuous, we show that the problem can be solved in polynomial time.
For all the settings we show that, assuming SETH, no truly subquadratic algorithm can exist.
\end{abstract}

\section{Introduction}

A \emph{curve} is any sequence of points in $\mathbb{R}^d$ where consecutive points are connected by their line segment. 
Curves may be used to model a variety of real-world input such as trajectories~\cite{takeuchi2021frechet}, handwriting~\cite{sriraghavendra2007frechet, zheng2008algorithm} and even strings~\cite{bringmann2015quadratic}.
Curves in $\mathbb{R}^1$ may be seen as \emph{time series} which model data such as music samples~\cite{serra2008chroma}, the financial market~\cite{taylor2008modelling} and seismologic data~\cite{yilmaz2001seismic}. 
A common way to analyse data that can be modeled as curves is to deploy a curve similarity measure, which for any pair of curves series $(P, Q)$ reports a real number (where the number is lower the more `similar' $P$ and $Q$ are).
Such similarity measures are a building block for common analysis techniques such as clustering~\cite{driemel2016clustering, wang2018time}, classification~\cite{adistambha2008motion,jeong2011weighted, kate2016using} or simplification~\cite{agarwal2005near,cheng2023curve,van2018global}.
The two most popular similarity measures for curve analysis are the Fr\'{e}chet distance and the Dynamic Time Warping (DTW) distance. 
The discrete Fr\'{e}chet distance for two curves $P = (p_1,\ldots,p_n)$ and $Q = (q_1,\ldots,q_m)$ is illustrated as follows. Imagine a dog walking along $Q$ and its owner walking along $P$. Both owner and dog start at the beginning of their curves, and in each step the owner may stay in place or jump to the next point along $P$ and the dog may stay in place or jump to the next vertex along $Q$, until both of them have reached the end of their curves. 
Intuitively, the \frechet distance is the minimal length of the leash between the dog and its owner. 
The DTW distance is defined analogously but sums over all leash lengths instead.

\mypar{Continuous distance measures.} The downside of discrete distance measures is that they are highly dependent on the sample rate of the curve. 
Both 
distance measures 
can be made continuous by defining a traversal as continuous monotone functions $f : [0, 1] \to P$ and $g : [0, 1] \to Q$ which start and end at the respective start and end of the curve. 
Given a traversal $(f, g)$, the cost of a traversal under the Fr\'{e}chet  distance is then the maximum over all $t \in [0, 1]$ of the distance between the dog and its owner at time $t$.
The DTW distance may also be made continuous, by defining the cost of a traversal as the integral over $[0, 1]$ of the distance between the dog and its owner. 
However, such a direct translation from discrete to continuous traversals invites degenerate behavior. 
To avoid such degeneracies, Buchin~\cite{buchin2007computability} proposed several variants of continuous DTW distances (originally called \emph{average Fr\'{e}chet} distance) that each penalise the speed of the dog and its owner.

\mypar{Computing Fr\'{e}chet and dynamic time warping distance.} The time complexity of computing the Fr\'{e}chet and DTW distance of two static curves is well-understood: 
Given two curves $P$ and $Q$ of length $n$ and $m$, we can compute their DTW distance and discrete  Fr\'{e}chet distance in time $O(nm)$ by dynamic programming approaches. This running time is almost the best known, up to mild improvements~\cite{DBLP:journals/siamcomp/AgarwalAKS14,GoldS18}. 
For continuous \frechet distance the best known algorithms run in slightly superquadratic time~\cite{buchin2014four}.
Bringmann~\cite{bringmann2014walking} showed that, conditioned on SETH, one cannot compute even a $(1 + \varepsilon)$-approximation of the Fr\'{e}chet distance between curves in $\mathbb{R}^2$ under the $L_1, L_2, L_\infty$ metric faster than $\Omega( (nm)^{1 - \delta})$ time for any $\delta > 0$. 
This lower bound was extended by Bringmann and Mulzer~\cite{bringmann2016approximability} to intersecting curves in $\mathbb{R}^1$. 
Buchin, Ophelders and Speckmann~\cite{buchin2019seth} show this lower bound for pairwise disjoint planar curves in $\mathbb{R}^2$ and intersecting curves in $\mathbb{R}^1$.
For the DTW distance similar SETH- or OVH-conditioned lower bounds exist~\cite{abboud2015tight,bringmann2015quadratic}  (even for curves in~$\mathbb{R}^1$).
Eor the DTW distanc, even if we relax the goal to a constant-factor approximation no algorithm running in truly subquadratic time is known (see \cite{kuszmaul2019dynamic} for polynomial-factor approximation algorithms and approximation algorithms for restricted input models), though for this approximate setting conditional lower bounds are still amiss.
Recently, a subquadratic constant-approximation for the continuous \frechet distance was presented by van der Horst et al.~\cite{DBLP:conf/soda/HorstKOS23}.

\mypar{Downsides of existing similarity measures.} The existing curve similarity measures each have their corresponding drawback:
The Fr\'{e}chet distance is not robust versus outliers. 
The discrete DTW distance is heavily dependent on the sampling rate. %
The continuous DTW variants are
robust to outliers, but they are difficult to compute~\cite{buchinCont}. 

\mypar{A measure for outliers.} We present a new curve similarity measure, specifically designed for computing similarities between curves under outliers, by applying a natural thresholding technique. Our algorithmic input receives in addition to $P$ and $Q$ some threshold $\rho$. If there exists a traversal of $P$ and $Q$ where at all times the dog and its owner are within distance $\rho$ then their distance is $0$. Otherwise, the distance is the minimal time spent where the dog and its owner are more than $\rho$ apart. 
Formally, we define a threshold function $\theta_\rho(p, q)$ which is $1$ whenever $d(p, q) > \rho$ and zero otherwise. Given a traversal $(f, g)$, the cost of the traversal is then the integral over $[0, 1]$ of $\theta_\rho(f(t), g(t))$. 
To avoid the same degenerate behavior as under the continuous DTW distance, we restrict the speeds of the dog and its owner. In particular, we create an asymmetric measure where we set the speed of the human to be constant and upper bound the speed of the dog by some input variable $s$.

\section{Preliminaries}

We consider $\mathbb{R}^d$, for constant $d$, where for two points $p, q \in \mathbb{R}^d$ we denote by $d(p, q)$ their Euclidean distance. 
We define a \emph{curve} as any ordered set of points (vertices) in $\mathbb{R}^d$.
When sequential vertices are connected by their line segment, this forms a polyline. 
For any curve $P$, $S$ is a subcurve of $P$ whenever its vertices are a subsequence of those of $P$. 
As algorithmic input we receive two curves $P = (p_1, p_2, \ldots, p_n)$ and $Q = (q_1, q_2, \ldots q_m)$.
By $|P|$ and $|Q|$ we denote their Euclidean length (the sum of the length of all the edges in the polyline). 
By $P[t]$ ($Q[t])$ we denote a continuous constant-speed traversal of $P$ ($Q$). That is, $P[t]$ is a continuous, monotone function for $t \in [1, n]$ to the points on the polyline defined by $P$ where $P[1] = p_1$, $P[n] = p_n$ and for all $t, t' \in [1, n]$ the distance between $P[t]$ and $P[t']$ (along $P$) is equal to $\frac{|P|}{|t - t'|}$.

\mypar{Discrete walks and reparametrizations.} Given two curves $P$ and $Q$, we define discrete walks and reparametrizations. 
First, consider the $n \times m$ integer lattice embedded in $\mathbb{R}^2$. 
We can construct a graph $G_{nm}$ over this lattice where the vertices are all lattice points and two lattice points $l_1, l_2$ share an edge whenever $d(l_1, l_2) \leq \sqrt{2}$. 

\begin{definition}
    \label{def:walks}
    For curves $P$ and $Q$, a \emph{discrete reparametrization} $F$ is any walk in $G_{nm}$ from $(1, 1)$ to $(n, m)$.
    $F$ is a curve in $\mathbb{R}^2$ and it is $x$-monotone whenever its embedding is. 
    The \emph{speed} $\sigma(F)$ is the size $|S|$ for the largest horizontal or vertical subcurve $S \subseteq F$. 
\end{definition}

\begin{definition}
    For a curve  $P$, a reparametrization of $P$ is a differentiable (monotone) function $\alpha : [0, 1] \to [1, |P|]$  where $\alpha(0) = 1$ and $\alpha(1) = |P|$. 
    For any two curves $(P,Q)$ with reparametrizations $(\alpha, \beta)$ we define the \emph{speed} as: $\sigma(\alpha, \beta) := \max\limits_{t \in [0, 1] }  \max \{ \alpha'(t), \beta'(t) \}$. 
\end{definition}

\begin{definition}
    \label{def:onesided}
    For a curve  $P$, a one-sided reparametrization (OSR) from $P$ to $Q$ is a differentiable (monotone) function $\gamma : [1, |P|] \to [1, |Q|]$ where $\gamma(1) = 1$ and $\gamma(|P|) = |Q|$. 
    For any two curves $(P,Q)$ with reparametrization $\gamma$ we define the \emph{speed} as $\sigma(\gamma) := \max\limits_{t \in [1, |P|] } | \gamma'(t) |$. 
\end{definition}

\mypar{Introducing speed restrictions to overcome degeneracies.} As discussed above any sensible definition of continuous DTW distance needs to somehow restrict the \emph{speed} of the reparamertization.
There have been several attempts at this. %
Buchin~\cite[Chapter 6]{buchin2007computability} proposed three ways of penalising using near-infinite speed by incorporating a multiplicative penalty for the speed in the definition of continuous DTW. 
Buchin, Nusser and Wong~\cite{buchinCont} present the current best algorithm for computing the first variant in $\mathbb{R}^1$, with a running time of $O(n^5)$. 
More closely related to our upcoming definition is the third definition in~\cite{buchin2007computability} which is one-sided and therefore asymmetric.
Alternatively, the speed of the reparametrization may simply be capped by some variable $s$.

\mypar{Outlier detection: speed bounds and thresholds.} The \frechet distance and DTW are both distance measures that are often used to compare sampled (trajectory) data.
However, in practice sampled trajectory data often contains measurement errors, leading to outliers or detours in the data, which we would like to detect. For this consider two polylines $P$ and $Q$, where we think of $P$ as the intended trajectory and $Q$ as a potentially faulty sample. We want to determine whether $Q$ is an accurate sample for $P$ or detect outliers if there are any.
Unfortunately, as can be seen in Figure~\ref{fig:outlier}, both the \frechet distance as well as (continuous) DTW are not adequate measures for this task.

\begin{figure}
    \centering
    \includegraphics[width=\textwidth]{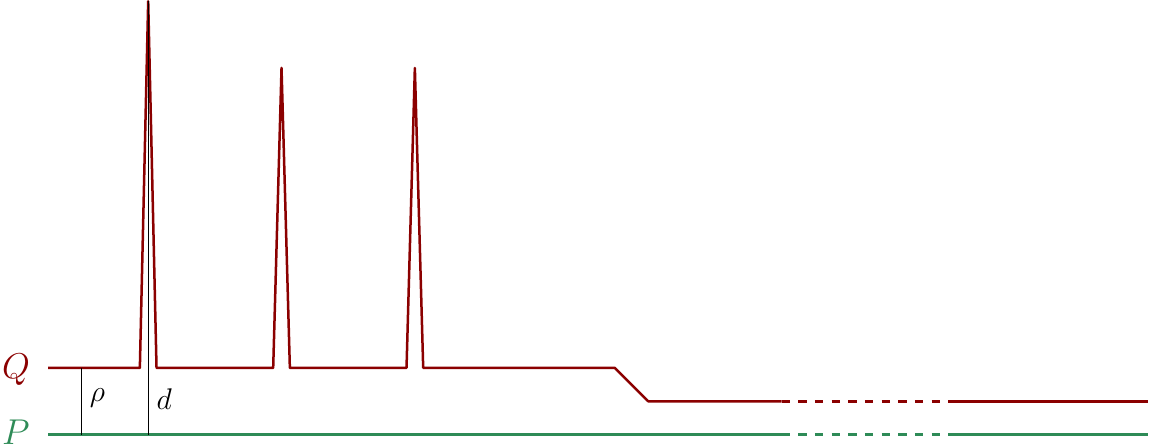}
    \caption{An intended trajectory $P$ and a faulty sample $Q$ of it. The \frechet distance between $P$ and $Q$ is $d$ and captures the first detour, but fails to capture the others. Continuous DTW, even with a speed bound, cannot distinguish $Q$ from a copy of $P$ translated by $\rho$ if the right part is sufficiently long. Barking distance with barking radius $\rho$ however captures all three detours.}
    \label{fig:outlier}
\end{figure}

\mypar{Defining Barking Distance.} The barking distance stems from the following illustration, which is again dog-based:\footnote{This illustration is inspired by a dog that some of the authors met while on a hike.} assume you are hiking with constant speed along a curve $P$. A dog is running at bounded speed on a curve $Q$, constantly barking at you. However, the dogs barks can only be heard within radius $\rho \in \mathbb{R}$. 
The dog tries to optimize its route in order to maximize the time you hear it. This maximum time is the barking distance of $P$ to $Q$. 
Formally, for $\rho \in \mathbb{R}$ we define the \emph{threshold} function as follows:
\[   
\theta_\rho(p,q) = 
     \begin{cases}
       1 &\quad\text{if } d(p, q) > \rho \\
       0 &\quad \text{ otherwise}.
     \end{cases}
\]

Just as for \frechet and DTW distance, we now define our distance measure in the continuous and discrete variants.
Observe that our definition includes a speed constraint $s$, to avoid the aforementioned degeneracy where one avoids detecting outliers by traversing through them at high speed. 
We choose a one-sided definition.
Like this, our definition stays close to the asymmetrical behavior of our hiker and dog.
Additionally, for detecting outliers, an asymmetric measure may also be more desirable since the input is inherently asymmetric.

\begin{definition}
    \label{def:discrete}
    For curves $P$ and $Q$, denote by $\mathbb{F}$ the set of all pairs of discrete $x$-monotone reparametrizations of $(P, Q)$. 
    For any $\rho, s \in \mathbb{R}$, the discrete barking distance is defined as:
\[
\mathbb{D}^s_B(P, Q) =  \min_{
\substack{ F \in \mathbb{F}  \\  \sigma(F) \leq s}} \sum_{(i, j) \in F} \theta_\rho(p_i, q_j).
\]
\end{definition}

\begin{definition}
    For curves $P$ and $Q$, denote by $\mathbb{G}$ the set of all one-sided reparametrizations from $P$ to $Q$. 
 For any $\rho, s \in \mathbb{R}$, the continuous barking distance is defined as:
\[
\mathcal{D}_B^s(P, Q) =  \min\limits_{ \substack{ \gamma \in \mathbb{G}  \\  \sigma(\gamma) \leq s}} \int_0^{|P|} \theta_\rho(P[t], Q[\gamma(t)]) \, dt. 
\]
\end{definition}

We also define a version of the barking distance that is in between the discrete and continuous barking distance.
In this \emph{semi}-discrete barking distance definition, the person has a discrete constant traversal of $P$ (spending $|(p_i, p_{i+1})|$ time at vertex $p_i$, before continuing to vertex $p_{i+1}$) 
whilst the traversal of the dog is continuous.  Formally:

\begin{definition}
    For curves $P$ and $Q$, denote by $\mathbb{G}$ the set of all one-sided reparametrizations from $P$ to $Q$. 
    Consider the  step function $\alpha : [0, |P|]$ where $\alpha[t]  := |P[1, i]| \textnormal{ whenever } t \in [|P[1, i]|, |P[1, i+1]| ]$.
 For any $\rho, s \in \mathbb{R}$, the semi-discrete barking distance is defined as:
\[
\mathcal{D}_B^s(P, Q) =  \min\limits_{ \substack{ \gamma  \in \mathbb{G}  \\  \sigma(\gamma) \leq s}} \int_0^{|P|} \theta_\rho(P[\alpha(t)], Q[\gamma(t)]) \, dt. 
\]
\end{definition}

\mypar{Contribution.} Our main contribution are efficient algorithms to compute the barking distance in several settings. In Section \ref{sec:discrete} we show that the discrete barking distance can be computed in time $O(nm\log s)$. In Section \ref{sec:semidiscrete} we show that the semi-discrete barking distance can be computed in time $O(nm\log (nm))$. In Section \ref{sec:continuous} we show that the continuous barking distance can be computed in time $O(n^4m^3\log (nm))$. While slower than the best known algorithm for DTW for time series \cite{buchinCont}, our algorithm works for curves in $\mathbb{R}^d$ for constant $d$. In Section \ref{sec:lower_bounds} we give conditional lower bounds showing that no truly subquadratic algorithm may exist to compute any variant, assuming the Strong Exponential Time Hypothesis (SETH). Hence, our algorithms for the discrete and semi-discrete barking distance are near-optimal. 
A feature of our reduction is that the inputs $(P, Q)$ are not only pairwise disjoint but also simple. 

\section{The discrete setting}\label{sec:discrete}
Let $G_{n,m} = (V,E)$ be a graph defined on top of
an $n\times m$ lattice in $\mathbb R^2$ where $v_{i,j} \in V$ is 
identified with the lattice point at coordinate $(i,j)$.
We find $(v_{i,j},v_{i',j'}) \in E$ 
with distinct $v_{i,j},v_{i',j'} \in V$
whenever $v_{i,j}$ and $v_{i',j'}$ are identified with points of the lattice
at distance $\leq \sqrt 2$.
We say that $v_{i',j'}$ is the 
\emph{southern}, \emph{south-western}, \emph{western}, \emph{north-western}, or \emph{northern} neighbor
of $v_{i,j}$ if $v_{i',j'}$ lies in the corresponding cardinal direction in the lattice.

For $v_{i,j} \in V$ we set $w(v_{i,j}) = \theta(p_i,q_j)$,
with $p_i$ the $i$-th corner of $P$ and $q_j$ the $j$-th corner of $Q$.
Similarly, we set $w(\pi) = \sum_{a=0}^k w(v_{i_a,j_a})$ for a walk $\pi = (v_{i_1,j_1},\ldots,v_{i_k,j_k})$ in $G_{n,m}$.
We say that $\pi$ is \emph{monotone} if $j_a \leq j_{a'}$ whenever $a \leq a'$ and 
we define the \emph{length} of $\pi$ as $|\pi|$, i.e., the number of vertices in the walk.
A sub-walk of $\pi$ is said to be 
\emph{horizontal} if all its vertices correspond to lattice points with the same $y$-coordinate and
\emph{vertical} if all its vertices correspond to lattice points with the same $x$-coordinate.
Moreover, we say that $\pi$ has \emph{speed} $s$ if the longest horizontal or vertical sub-walk of $\pi$ 
has length at most $s$.
Let $\Pi(s,\rho)$ be the set of all monotone walks in $G_{n,m}$ starting at $v_{1,1}$ and ending at $v_{n,m}$ 
with speed $s$ and weight function depending on the threshold $\rho$.
The next observation now follows with Definitions~\ref{def:walks} and~\ref{def:discrete}.

\begin{observation}
    \label{obs:shortpaths}
    Given two polygonal curves $P$ and $Q$, a threshold $\rho$, and a speed bound $s$,
    let $G_{n,m}$ be defined as above,
    then $w(\pi) = \mathbb{D}^s_B(P, Q)$ for any $\pi \in \Pi(s,\rho)$ of minimum weight.
\end{observation}

By Observation~\ref{obs:shortpaths} we can restrict our attention to monotone paths from $v_{1,1}$
to $v_{n,m}$ that have speed at most $s$ and are of minimum weight.
Our strategy is to compute for each vertex $v_{i,j} \in V$
the weight of such a path from $v_{1,1}$ and to $v_{i,j}$.
Our computation will proceed in $n$ rounds, where 
in each round we consider the $m$ vertices of column $j$.
The challenge is to compute the length of a minimum weight monotone path of speed $s$ 
in time $O(\log s)$.

Let $R_i(j_1,j_2)$ be the weight of path $(v_{i,j_1+1},v_{i,j_1+2},\ldots,v_{i,j_2})$ and
$C_j(i_1,i_2)$ be the weight of path $(v_{i_1+1,j},v_{i_1+2,j},\ldots,v_{i_2,j})$.
Observe, that these values can be computed in constant time
if we have arrays containing at position $i$ the length of a path from
the first element in the row or column to the $i$-th element of the row or column.
For each row and column and taking either side as the starting vertex.
We precompute these arrays for all rows at the beginning and 
for each column only when we process this column in the computation.

Let $F_\delta(i,j)$ with $\delta \in D = \{\uparrow,\nearrow,\rightarrow,\searrow,\downarrow\}$ be 
the minimum weight of a monotone path of speed $s$ from $v_{1,1}$ to $v_{i,j}$
where the vertex preceding $v_{i,j}$ on the path is the
southern, south-western, western, north-western, or northern neighbor of $v_{i,j}$, respectively.
We set $F_\delta(i,j) = \infty$ 
if $v_{i,j}$ cannot be reached with any monotone path of speed $s$ from $v_{1,1}$.
We then compute the minimum weight monotone path of speed $s$ from $v_{1,1}$ to $v_{i,j}$ as
$F(i,j) = \min \{F_\delta(i,j) \mid \delta \in D\}$.
To compute $F(i,j)$ from left to right along the columns
we maintain the relevant 
minima of paths $F_\delta$ ending at vertices around $v_{i,j}$ 
for each row and for the current column in separate heaps.
Moreover, instead of updating the weights of all heap-elements explicitly for each $v_{i,j}$, 
we precompute the lengths of paths starting at the beginning or end of a row or column.
From this we can in constant time compute the necessary offsets.
The runtime of $O(nm\log s)$ then follows as every of the $O(nm)$ elements gets only inserted and deleted 
from some min-heap a constant number of times and
at no point any min-heap contains more than $s$ elements.

For the following proof we rewrite $F_d(i,j)$ as a recurrence taking the speed-bound $s$ into account for $j > 1$. Recall that  $w(v_{i,j})$ contributes to the values of $C_j$ and $R_i$.
\[    
    F_d(i,j) = 
    \begin{cases}
        \min\{C_j(i-k,i) + F_\delta(i-k,j) \mid \delta \in \{\nearrow,\rightarrow,\searrow\} \land k \in [1,s]\} & \text{if } d = \uparrow\\
        F(i-1,j-1) + w(v_{i-1,j-1}) & \text{if } d = \nearrow \\
        \min\{R_i(j-k,j) + F_\delta(i,j-k) \mid \delta \in \{\uparrow, \nearrow, \searrow, \downarrow\} \land k \in [1,s]\} & \text{if } d = \rightarrow \\
        F(i+1,j+1) + w(v_{i+1,j+1}) & \text{if } d = \searrow\\
        \min\{C_j(i+k,i) + F_\delta(i+k,j) \mid \delta \in \{\nearrow,\rightarrow,\searrow\} \land k \in [1,s]\} & \text{if } d = \downarrow\\
    \end{cases}
\]

\begin{restatable}{theorem}{discretethm}
Given two polygonal curves $P$ and $Q$ with $n$ and $m$ vertices, respectively, the discrete Barking distance of $P$ to $Q$ can be computed in time $O(nm(\log s))$
where $s$ is the speed bound or time $O(nm\log(nm))$ if $s > n$ or $s > m$.
\end{restatable}
\begin{proof}
For the first column, i.e., $j=1$, we can compute $F(i,j)$ as follows.
Clearly, $F(1,1) = w(v_{1,1}) = \theta_\rho(p_1,q_1)$.
We set $F(i,1) = \infty$ for all $i > s$.
Finally, we find that the remaining entries $F(i,1)$ with $i \in [2,s]$ 
in a bottom-up traversal as the values $C_1(1,i)$.
We conclude this step by initializing a min-heap $H_i$ for each row $i$ 
containing vertex $v_{i,1}$ as its sole element and $F(i,1) = F_\rightarrow(i,1)$ as the key.

Assume now that we want to compute the entries $F_d(i,j)$ for $i \in [1,m]$ where 
all entries $F_d(i,j')$ with $j' < j$ are already computed and
for row $i$ we have a min-heap $H_{i}$ containing all $F_\rightarrow(i,j-k)$ for $k \in [1,s]$ ordered
by key $F_\rightarrow(i,j-k) + R_i(j-k,j-1)$.
From this information we can for each $i$ immediately compute 
$F_\rightarrow(i,j)$ as the minimum of $H_i$, say $v_{i,j'}$ plus $R_i(j-j',j)$.
We then update $H_i$ by deleting all entries for $v_{i,j-s}$ and then
inserting $(v_{i,j},F_\uparrow(i,j))$, $(v_{1,j},F_\downarrow(i,j))$, $(v_{1,j},F_\searrow(i,j))$,
and $(v_{1,j},F_\nearrow(i,j))$ 
using as key for comparison $F_{\cdot}(i,j) + R_i(j-k,j)$ in the insertion.
Note that since for all elements already present in $H_i$ their keys change only by $w(v_{i,j})$ 
and hence their order remains the same.
Moreover, we can directly compute $F_\nearrow(i,j)$ and $F_\searrow(i,j)$ for each $i$.

It remains to compute $F_\uparrow(i,j)$ and $F_\downarrow(i,j)$ for each $i\in[1,m]$ in column $j$.
We describe how to compute $F_\uparrow(i,j)$, $F_\downarrow(i,j)$ can be computed symmetrically.
We start from $v_{1,j}$.
Clearly, $F_\uparrow(1,j) = \infty$.
We also initialize a min-heap $H$ and insert $(v_{1,j},F_\rightarrow(1,j))$, $(v_{1,j},F_\searrow(1,j))$,
and $(v_{1,j},F_\nearrow(1,j)=\infty)$
where the second element is used as key.
Assume that we now want to compute $F_\uparrow(i,j)$ and 
that we have a heap $H$ containing for $k \in [1,s]$ the elements
$(v_{i-k,j},F_\rightarrow(i-k,j))$, $(v_{i-k,j},F_\searrow(i-k,j))$, and $(v_{i-k,j},F_\nearrow(i-k,j))$
ordered by key $F_{\cdot}(i-k,j) + C_j(i-k,i-1)$.
This can be done as for the row by just extracting the minimum element from the heap $H$, 
say $(v_{i',j},F_{\delta}(i',j))$,
and setting $F_\uparrow(i,j) = F_{\delta}(i',j) + C_j(i',i)$.
We update the heap as in the case of $H_i$, with the only difference being that we need to insert the three elements $(v_{i,j},F_\rightarrow(i,j))$, $(v_{1,j},F_\searrow(i,j))$,
and $(v_{1,j},F_\nearrow(i,j))$

Correctness follows since the algorithm computes directly the above recurrence.
Moreover, since every element vertex and partial weight combination gets deleted and inserted at most once 
from some heap over the whole computation and no heap contains more than $3s$ elements at a time,
we obtain the claimed running time of $O(nm\log s)$.
Note, that if $s > m$ or $s > n$ we obtain a runtime of $O(nm(\log(m) + \log(n)))$ since 
again never more than $O(s)$ elements are contained in a heap and no more than $O(nm)$ elements can be inserted or deleted.    
\end{proof}

\section{The semi-discrete setting}\label{sec:semidiscrete}

We show an $O(nm \log(nm))$ time algorithm to compute the semi-discrete barking distance. 
To this end, we define a few critical variables and functions:

\begin{definition}
Let $P$ and $Q$ be any two curves with $n$ and $m$ vertices respectively and let $\rho, s \in \mathbb{R}$ be fixed. 
We define for each $i \in [n]$:
\begin{itemize}
    \item By $\Lambda_i$ the vertical line with $x$-coordinate $x_i = |P[1, i]|$ in $[|P|] \times [|Q|]$.
    \item For any $(\lambda_{i-1}, \lambda_{i}) \in \Lambda_{i-1} \times \Lambda_{i}$ we define the minimal cost of any path from $\lambda_{i-1}$ to $\lambda_{i}$. 
    Formally let $G_i$ denote all differentiable functions $\gamma : [x_{i-1}, x_i] \to [0, |Q|]$. 
We define:
    \[
    f(\lambda_{i-1}, \lambda_{i}) := \min\limits_{ 
    \substack{ \gamma \in G_{i}, \gamma'(t) \leq s \\ 
                f(x_{i-1}) = \lambda_{i}, f(x_i) = \lambda_{i+1}}  } \int_{x_{i-1}}^{ x_i } \theta_\rho(P[x_{i-1}], Q[f(t)]) \, dt.  
    \]
\item Finally, we define $F_i(\lambda_i)$ as the minimum cost of all speed-bounded traversals where the dog ends at $\lambda_i \in \Lambda_i$.
In other words, we have
$F_i(\lambda_i):=\min_{\lambda_{i-1}}(F_{i-1}(\lambda_{i-1})+f(\lambda_{i-1},\lambda_i)).$
\end{itemize}

\end{definition}

Note that the barking distance from $P$ to $Q$ is now just the value of $F_n(|Q|)$.

\mypar{Barkable Space Diagram.} To compute these values, we define the \emph{Barkable Space Diagram} (BSD).\footnote{As a geometric object this diagram is the same as the \emph{free space diagram} for \frechet distance. We use a different name to highlight that the paths through the diagram have different rules.}
The BSD is a graph over $[0, |P|] \times [0, |Q|] \subset \mathbb{R}^2$.
The $x$-axis corresponds to the location of the hiker in its (fixed) traversal of $P$, and the $y$-axis corresponds to the location of the dog on $Q$ (parameterized according to $Q[t]$).
In particular, a point $(x,y)$ corresponds to the pair $(P[x],Q[y])$.
For any $(x, y) \in [0, |P|] \times [0, |Q|]$ we define $BSD[x, y] := \theta_\rho(P[x], Q[y])$.
We denote $B_\rho = \{  (x, y) \in [0, |P|] \times [0, |Q|] \mid BSD[x, y] = 1 \}$ and call this set the \emph{obstacles}. 

A traversal of the dog (i.e. a function $\gamma \in \mathbb{G}$) is a curve $\gamma$ through the BSD.
Furthermore, the barking distance of $\gamma$ is the total time spent in obstacles, i.e., the length of the projection onto the $x$-axis of the curve $\gamma \cap B_{\rho}$. 
The speed bound for the dog translates into a maximal and minimal slope for $\gamma$.

Our BSD naturally partitions into vertical \emph{slabs} in between lines $\Lambda_{i-1}$ and $\Lambda_i$. 
For any curve $\gamma \in [0, |P|] \times [0, |Q|]$, we can extract the values $\lambda_i = \min \{ \gamma \cap \Lambda_i \}$. 
We determine the values $\lambda_i$ corresponding to the optimal curve $\gamma$ by reasoning about the `shape' of the BSD:

\begin{restatable}{lemma}{bsdsemidisclemma}
\label{lem:bsd_semi_discrete}
For any $i \in [n]$, the set $B_\rho \cap  \left( [x_{i-1}, x_i] \times [0, |Q|] \right)$ are $O(m)$ axis-parallel rectangles whose vertical boundaries are incident to $\Lambda_{i-1}$ and $\Lambda_{i}$. 
\end{restatable}
\begin{proof}
Fix some $p_i\in P$. For any $q\in Q$, consider the distance $d(p_i, q)$.
This partitions $Q$ into maximal subcurves where for each subcurve $S$, all $q \in S$ have either $d(p_i, q) \leq \rho$ or $d(p_i, q) > \rho$. 
Each maximal subcurve $S = Q[y_1, y_2]$ where for all $q \in S$, $d(p_i, q) > \rho$ corresponds uniquely to a connected component in $B_\rho \cap (\cap) [x_{i-1}, x_i]  \times [0, |Q|])$ are $O(m)$. This connected component has its vertical boundaries coinciding with $\Lambda_i$ and $\Lambda_{i+1}$. Its vertical boundaries are horizontal segments of height $y_1$ and $y_2$. 

Note that each edge of $Q$ can contain at most one endpoint of a maximal subcurve $S$ where for all $q \in S$, $d(p_i, q) > \rho$.
This bounds the number of connected components in $B_\rho \cap ([x_{i-1}, x_i] \times [0, |Q|])$ to be $O(m)$. 
\end{proof}

\subsection{\boldmath Understanding the function \texorpdfstring{$f(\lambda_{i-1},\lambda_i)$}{f}}
The function $f(\lambda_{i-1},\lambda_i)$ takes values in $\Lambda_{i-1} \times \Lambda_i$ and returns a cost, which is a value in $\mathbb{R}_{\geq 0}$. In particular, in order to understand the behaviour of the function $f$, we plot a new graph where we think of $\lambda_i \in \Lambda_i$ to be the $x$-axis and $\lambda_{i-1} \in \Lambda_{i-1}$ to be the $y$-axis.\footnote{This choice allows us to read off a univariate function that appears later as a function over the $x$-axis.}

\begin{restatable}{lemma}{fsymmetriclemma}
\label{lem:f_symmetric}
For all $i \in [n]$, the function $f(\lambda_{i-1},\lambda_i)$ is symmetric, i.e., $f(a,b)=f(b,a)$.
\end{restatable}
\begin{proof}
Consider a (partial) traversal of $Q$ which, whilst the human remains stationary at $p_i$, starts at $\lambda_{i-1}=a$ and ends at $\lambda_i=b$.
Any such traversal, in the BSD, is an $x$-monotone curve $\gamma$ contained within the vertical slab corresponding to $p_i$.
As the obstacles are rectangles by Lemma \ref{lem:bsd_semi_discrete}, mirroring this curve gives an optimal curve from $\lambda_{i-1}=b$ to $\lambda_i=a$.
\end{proof}

Fix now a slab with boundaries $\Lambda_{i-1}$ and $\Lambda_i$.
Consider the obstacles $O_1,\ldots, O_k$ in the slab.
We may note for each $O_j$ by $b_j$ and $t_j$ the height of the bottom and top boundary of $O_j$, respectively. This defines a natural grid structure on the plot of $f$ over $\Lambda_{i-1} \times \Lambda_i$.

We say that a point $\lambda_{i-1}$ on the vertical line $\Lambda_{i-1}$, is \emph{covered} (with respect to the slab $[x_{i-1}, x_i] \times [0, |Q|]$) if it lies on the boundary of some obstacle $O_j$, that is, if we have $b_j\leq \lambda_{i-1}\leq t_j$ for some $j$. Otherwise, we say that it is \emph{uncovered} (with respect to the current slab). Analogously, say that $\lambda_i\in\Lambda_i$ is covered (with respect to the slab $[x_{i-1}, x_i] \times [0, |Q|]$) if $b_j\leq \lambda_{i-1}\leq t_j$ for some $j$ and uncovered otherwise.
In the following, the slab under consideration is always fixed, so we just say that a point is covered or uncovered.
This defines three different types of rectangles in the grid structure on the plot of $f$: a point $(\lambda_{i-1},\lambda_i)$ is \emph{completely covered} if both $\lambda_i$ and $\lambda_{i-1}$ are covered, it is \emph{half-covered} if one of the is covered but the other is uncovered, and it is \emph{completely uncovered} if both $\lambda_i$ and $\lambda_{i-1}$ are uncovered.

\begin{restatable}{lemma}{bivariatediagramlemma}
\label{lem:f_bivariate_diagram}
For all vertices $p_i\in P$, the function $f(\lambda_{i-1},\lambda_i)$ is a piecewise linear function.
Further, the pieces can be described in the plot of $f$ as follows:
\begin{itemize}
    \item[a)] there are two lines $\ell_1$ and $\ell_2$ of slope $1$ such that for all points $(\lambda_{i-1},\lambda_i)$ below $\ell_1$ or above $\ell_2$ we have $f(\lambda_{i-1},\lambda_i)=\infty$;
    \item[b)] in between 
    the pieces form an $(O(m)\times O(m))$-grid (truncated by $\ell_1$ and $\ell_2$) where 
    the completely covered 
    cells crossing the diagonal are further divided by segments of slope $-1$.
\end{itemize}
\end{restatable}
\begin{proof}
We claim that there is always an optimal traversal (i.e., a traversal of $Q$ whose cost corresponds to a value of $f(\lambda_{i-1},\lambda_i)$) whose curve $\gamma$ consists of at most two line segments of maximal (or minimal) slope and at most one horizontal line segment.
Indeed, any segment of non-optimal slope can be replaced by a segment of optimal slope and a horizontal segment. As the borders of the obstacles are horizontal segments by Lemma \ref{lem:bsd_semi_discrete}, this can only improve the time spent in an obstacle.
Similarly, any two segments of positive (negative) slope with a horizontal segment between them can be replaced by a single segment of positive (negative slope) with a horizontal segment afterwards.
Finally, as each obstacle is a rectangle, no optimal path will zig-zag, thus we may indeed assume that $\gamma$ consists of at most three segments as claimed above.
In particular, note that such an optimal curve $\gamma$ will only have a horizontal segment in an obstacle if $\gamma$ lies completely in this obstacle.
This already implies that the function is piecewise linear: the time an optimal segment stays in an obstacle is proportional to the vertical distance of its start or endpoint to the boundary of the obstacle.

It remains to prove the claims on the pieces. Claim a) follows immediately from the speed bound, as from a point $\lambda_{i-1}\in\Lambda_{i-1}$ we can only reach points $\lambda_i\in\Lambda_i$ with $\lambda_{i-1}-\ell\leq \lambda_i\leq \lambda_{i-1}+\ell$, where the $\ell$ is a number that depends on the speed bound.
For claim b), consider two points $\lambda_{i-1}\in\Lambda_{i-1}$ and $\lambda_i\in\Lambda_i$, where the cost of the optimal path between them by definition is $f(\lambda_{i-1},\lambda_i)$. We will distinguish the cases whether $\lambda_{i-1}$ and $\lambda_i$ are covered or uncovered.

Let us start with the case that both $\lambda_{i-1}$ and $\lambda_i$ are uncovered. Then, for any $x\in\Lambda_{i-1}$ and $y\in\Lambda_i$ in the same uncovered regions we have that if $y$ is reachable from $x$, then $c(x,y)$, the cost of the path from $x$ to $y$, equals $c(\lambda_{i-1},\lambda_i)$: as by the above observations any optimal path traverses obstacles with maximum speed, this cost is proportional to the width of the obstacles between the two uncovered regions.

Let us now assume that $\lambda_{i-1}$ is covered by some obstacle $O_j=O$ and that $\lambda_{i}$ is either also covered by $O$ or lies in an uncovered region adjacent to the region covered by $O$. We distinguish three different subcases, depending on the width of the obstacle, starting with the largest and ending with the smallest. Recall that $t$ and $b$ denote the topmost and bottommost point of the left boundary of the obstacle $O$. Further denote by $t'$ the lowest point on the right boundary that is reachable from $t$ and by $b'$ the highest point reachable from $b$.

\textbf{Case 1:} there is a point on the right boundary of the obstacle that can neither be reached from $t$ nor from $b$. In particular we have $b<b'<t'<t$. For an illustration of this case we refer to Figure \ref{fig:case3}(a). In any cell of the form $[b',t']\times I$ and, symmetrically, $I\times [b',t']$, the function $f$ is constant at the highest possible cost $c:=x_i-x_{i-1}$. The cell $[b,b']\times[b,b']$ is subdivided into two triangles by a segment of slope $-1$, where $f$ is constant $c$ on the upper triangle and linearly increasing with gradient $(1,1)$ on the lower triangle. Symmetrically, the cell $[t',t]\times[t',t]$ is also subdivided into two triangles, now with the lower one constant and the upper one linear with gradient $(-1,-1)$. Finally the cells of the form $[b,b']\times[\cdot, b]$ (and their symmetric variants) are linear with gradient $(1,0)$ (or the corresponding negative or swap of coordinates), whereas the cells $[\cdot, b]\times[\cdot, b]$ and $[t, \cdot]\times[t, \cdot]$ are constant 0.

\textbf{Case 2:} every point on the right boundary can be reached either from $t$ or from $b$, but some of them only from one of them. We thus have $b<t'<b'<t$. For the structure of $f$ in this case, we refer to Figure \ref{fig:case3}(b).

\textbf{Case 3:} every point on the right boundary can be reached both from $t$ and from $b$. We thus have $t'<b<t<b'$. For the structure of $f$ in this case, we refer to Figure \ref{fig:case3}(c).

The cases where $\lambda_{i-1}$ is covered and $\lambda_{i}$ is in a different region differ from the above cases by the addition of a value that is proportional to the width of the obstacles that are completely between. Finally, the cases where $\lambda_{i-1}$ is uncovered and $\lambda_{i}$ is covered follow from symmetry.
The size of the grid is thus bounded by the number of obstacles, which by Lemma~\ref{lem:bsd_semi_discrete} is $O(m)$.
\end{proof}

For an illustration of this, see Figure \ref{fig:semidiscrete_example}.
\begin{figure}
    \centering
    \includegraphics[page=1, width=\textwidth]{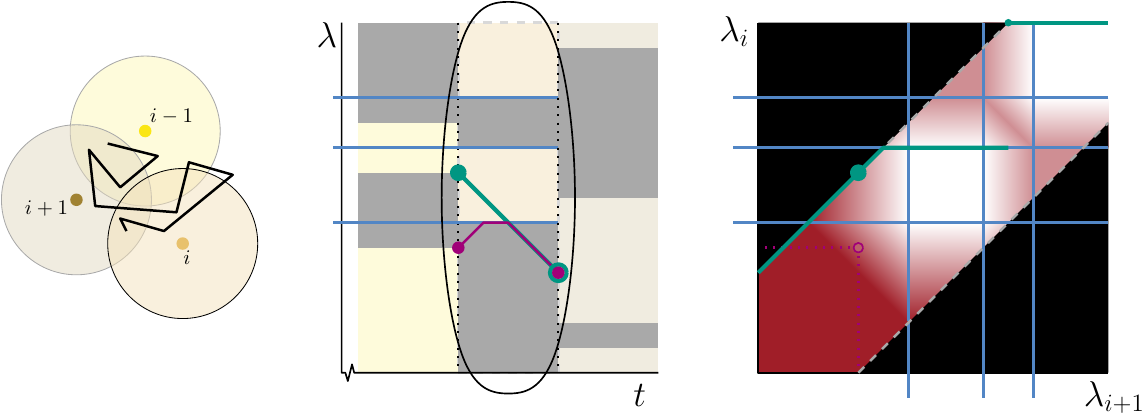}
    \caption{Two trajectories $P$ and $Q$ (left), the corresponding BSD (middle) and the plot of $f(\lambda_{i-1},\lambda_i)$ (right). In the plot of $f(\lambda_{i-1},\lambda_i)$ there is also the graph of $\Psi(\lambda_i)$.}
    \label{fig:semidiscrete_example}
\end{figure}
The lemma can be proven by a careful analysis of the possible paths $\gamma$. Illustrations of the behaviour of $f(\lambda_{i-1},\lambda_i)$ around differently sized obstacles can be found in Figure~\ref{fig:case3}. 

\begin{figure}
    \centering
    \includegraphics[page=1]{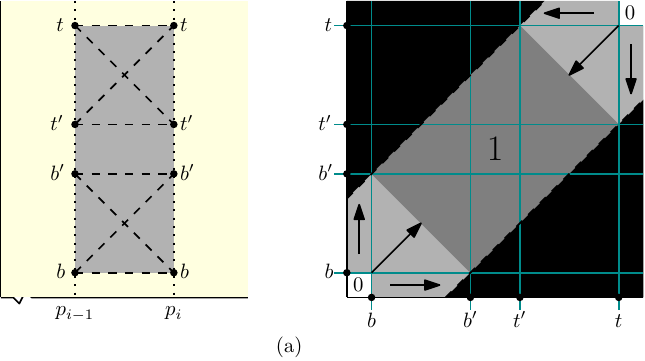}
    \vspace{1cm}
    
    \includegraphics[page=1]{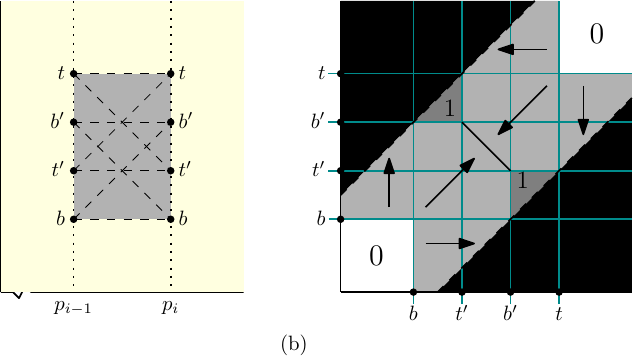}
    \vspace{1cm}
    
    \includegraphics[page=1]{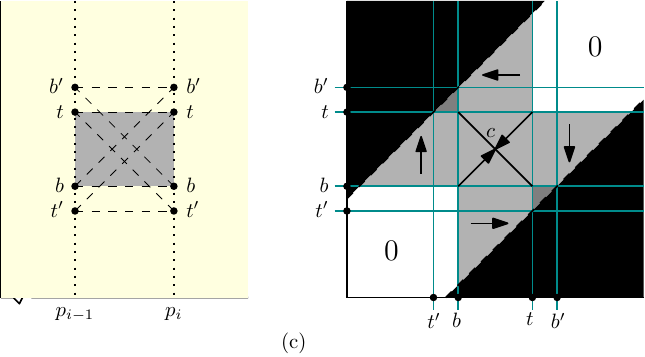}
    \caption{The graphs of $f$ for a large (a) medium (b) and small obstacle (c).}
    \label{fig:case3}
\end{figure}

Fixing the endpoint of the path $\gamma$ corresponds to fixing the value $\lambda_i=y$. This defines a univariate function $f(\lambda_{i-1})=f(\lambda_{i-1},y)$, which in the plot of $f$ corresponds to the restriction to a vertical line. Depending on whether $y$ is covered or uncovered, this line only passes through completely covered and half-covered or through half-covered and completely uncovered cells, where each cell defines an interval on the line. We call these intervals \emph{covered} if $\lambda_{i-1}$ is covered and \emph{uncovered} otherwise. The following now follows from Lemma \ref{lem:f_bivariate_diagram}.

\begin{lemma}\label{lem:f_univariate_structure}
For each $\lambda_{i}=y$, for the (univariate) function $f_y(\lambda_{i-1})=f(\lambda_{i-1},y)$ there are values $p_1\leq p_2\leq y\leq p_3\leq p_4$ such that
\begin{itemize}
    \item $p_1,\ldots, p_4$ correspond to coordinates of horizontal lines of the grid structure of $f(\lambda_{i-1},\lambda_i)$;
    \item $f_y(\lambda_{i-1})=\infty$ for $\lambda_{i-1}<p_1$ and $\lambda_{i-1}>p_4$;
    \item $f_y(\lambda_{i-1})$ is piecewise linear where all pieces have slope $-1$ on covered intervals and $0$ on uncovered intervals for $p_1\leq \lambda_{i-1}\leq p_2$;
    \item $f_y(\lambda_{i-1})$ is piecewise linear where all pieces have slope $1$ on covered intervals and $0$ on uncovered intervals for $p_3\leq \lambda_{i-1}\leq p_4$;
    \item $f_y(\lambda_{i-1})$ is linear with slope $1$ and then constant for $p_2\leq \lambda_{i-1}\leq y$;
    \item $f_y(\lambda_{i-1})$ is constant and then linear with slope $-1$ for $y\leq \lambda_{i-1}\leq p_3$.
\end{itemize}
Further, if $\lambda_{i-1}$ is uncovered, that is, not on the boundary of an obstacle, then we can choose $p_2=p_3$ and ignore the last two points. 
\end{lemma}

\subsection{Computing the barking time}

Our algorithm computes $F_i(\lambda_{i})$ from $F_{i-1}(\lambda_{i-1})$ and $f(\lambda_{i-1},\lambda_i)$. For this, we need to understand for each point $\lambda_i$ the point $\lambda_{i-1} \in \Lambda_{i-1}$ which is intersected by an optimal curve $\gamma \in G_i$. 
As there might be several optimal curves, we will always consider the highest one, that is, the one maximizing $\lambda_{i-1}$.
We define 
$\Psi(\lambda_i):=\max\{\text{argmin}_{\lambda_{i-1}}(F_{i-1}(\lambda_{i-1})+f(\lambda_{i-1},\lambda_i))\}.$

\begin{restatable}{lemma}{argminmonotonelemma}
\label{lem:argmin}
The function $\Psi(\lambda_i)$ is monotone.
\end{restatable}
\begin{proof}
Assume for the sake of contradiction that $\Psi(\lambda_i)$ is not monotone, that is, there are values $q_1<q_2$ for which $\Psi(q_1)=:r_1>r_2:=\Psi(q_2)$. That is, the highest optimal path $\gamma_1$ to $q_1$ passes through $r_1$ and the highest optimal path $\gamma_2$ to the higher $q_2$ goes through $r_2$, which is lower than $r_1$. In particular, the paths $\gamma_1$ and $\gamma_2$ cross in some point $x$. Now the cost of $\gamma_2$ from to $x$ has to be strictly smaller than the cost of $\gamma_1$ to $x$, otherwise we would get a path to $q_2$ with smaller or equal cost than $\gamma_2$ by following $\gamma_1$ until $x$ and only following $\gamma_2$ from $x$ to $q_2$, so $\gamma_2$ would not be the highest optimal path. But then the symmetric construction gives us a path to $q_1$ with strictly smaller cost than $\gamma_1$, which is a contradiction.
\end{proof}

\begin{restatable}{lemma}{fstructurelemma}
    \label{lem:Fstructure}
The function $F_i(\lambda_i)$ is a piecewise linear function with $O(im)$ parts.
\end{restatable}
\begin{proof}
We prove this by induction on $i$. For the base case we may assume $F_0(\lambda_0)=0$.
By Lemma \ref{lem:f_bivariate_diagram} $f(\lambda_{i-1},\lambda_i)$ is piecewise linear and its cells form (more or less) a grid structure. By the induction hypothesis $F(\lambda_{i-1})$ is piecewise linear and as it only depends on $\lambda_{i-1}$, we get that the function $F_{i-1}(\lambda_{i-1})+f(\lambda_{i-1},\lambda_i)$ is also piecewise linear and each of its pieces is defined by subdividing the pieces of $f(\lambda_{i-1},\lambda_i)$ by horizontal lines corresponding to the non-differentiable points of $F_{i-1}(\lambda_{i-1})$. In particular, for each $\lambda_i$ we have that $\Psi(\lambda_i)$ lies on the boundary of a piece.
This already implies that $F_i(\lambda_i)$ is piecewise linear, and it remains to bound the number of parts.

For this we consider the graph of $\Psi(\lambda_i)$. By Lemma \ref{lem:argmin}, $\Psi(\lambda_i)$ is monotone, so this graph is $xy$-monotone, and as it follows the boundaries of the pieces, it must thus consist of segments on $\ell_1$ and $\ell_2$ and horizontal segments with vertical jumps, i.e., outside of $\ell_1$ and $\ell_2$ it is the graph of a step function. Further, a new part of $F_i(\lambda_i)$ can only occur whenever the graph of $\Psi(\lambda_i)$ crosses a vertical boundary of a piece, makes a vertical jump or crosses the horizontal boundary of a piece while following $\ell_1$ or $\ell_2$. For the first type of changes we note that all the boundaries that were added with the addition of $F_{i-1}(\lambda_{i-1})$ are horizontal, so the number of vertical boundaries is still $O(m)$ as proven in Lemma \ref{lem:f_bivariate_diagram}. As for the number of vertical jumps, we note that every vertical jump is either on a vertical boundary or crosses a diagonal segment of slope $-1$, where each such segment can only be crossed once by monotonicity. As by Lemma \ref{lem:f_bivariate_diagram} the number of diagonal segments is also $O(m)$, these also contribute to $O(m)$ changes. Finally, the number of horizontal lines that can be crossed is $O(m)$ for the grid lines of $f(\lambda_{i-1},\lambda_i)$ and $O((i-1)m)$ for the lines corresponding to non-differentiable points of $F_{i-1}(\lambda_{i-1})$. It follows that the total number of parts is $O(im)$.
\end{proof}

From the above proof it follows that there are two types of non-differentiable points in $F_{i}(\lambda_i)$. The first type are the ones that were already present in $F_{i-1}(\lambda_{i-1})$, i.e., the ones corresponding to points where the graph of $\Psi(\lambda_i)$ goes along $\ell_1$ or $\ell_2$ and crosses a horizontal line induced by $F_{i-1}(\lambda_{i-1})$. The second type are the newly introduced ones, that is, the once defined by crossing a grid line of $f(\lambda_{i-1},\lambda_i)$ or a vertical jump. The arguments above show that the number of non-differentiable points of the second type is on $O(m)$. We call these points \emph{new breaking points}. We call the horizontal lines induced by the new breaking points of $F_{i-1}(\lambda_{i-1})$ as well as the first horizontal lines before and after them \emph{breaking lines}.%

\begin{restatable}{lemma}{monotonebreaklemma}
\label{lem:monotone_between_breakings}
Between any two new breaking points the function $F_i(\lambda_i)$ is monotonically increasing or decreasing and each piece has integer slope.
\end{restatable}
\begin{proof}
We prove this by induction on $i$. The base case $F_0(\lambda_0)$ is again trivial.
We first note that the only case where $F_i(\lambda_i)$ is not just a single segment between two new breaking points is if $\Psi(\lambda_i)$ follows $\ell_1$ or $\ell_2$. Consider such a part of $\Psi(\lambda_i)$, without loss of generality on $\ell_1$ and the horizontal lines it crosses. The crossed horizontal lines correspond to non-differentiable points $a_1<a_2<\ldots<a_k$ of $F_{i-1}(\lambda_{i-1})$ on $\Lambda_{i-1}$ from which to considered optimal paths in the current slab are segments of maximal slope. Denote by $b_j$ the endpoint of the segment starting at $a_j$. We first note that $F_{i-1}(a_1)\leq\ldots\leq F_{i-1}(a_k)$, as otherwise the segment of maximal slope would not be an optimal path. As the pieces of $F_{i-1}(\lambda_{i-1})$ have integer slopes by the induction hypothesis and $f(\lambda_{i-1},\lambda_i)=\alpha x+c$ for some constant $c$ and an $\alpha\in\{-1,0,1\}$, we thus already get that the slopes of the relevant pieces are in the set $\{-1,0,1\ldots\}$. It remains to show that there is no segment of slope $-1$. Assume for the sake of contradiction that $[F_i(b_j),F_i(b_{j+1})]$ has slope $-1$. This can only happen if $F_{i-1}(a_j)=F_{i-1}(a_{j+1})$ and $f(a_j,b_j)=f(a_{j+1},b_{j+1})-1$. In particular, $b_j$ and $b_{j+1}$ are uncovered. But then the segment from $a_{j+1}$ to $b_j$ gives a path to $b_j$ with smaller cost, which is a contradiction.
\end{proof}

\begin{corollary}\label{cor:order}
Let $a_1$ and $a_2$ be two consecutive non-differentiable points of $F_{i-1}(\lambda_{i-1})$ that induce non-differentiable points $b_1$ and $b_2$ in $F_i(\lambda_i)$. Then $b_1$ and $b_2$ are also consecutive.
\end{corollary}

\begin{restatable}{lemma}{psipathlemma}
\label{lem:Psi_path}
Whenever the graph of $\Psi(\lambda_i)$ follows a horizontal line, this line is either a grid line or a breaking line.
\end{restatable}
\begin{proof}
Fix a value $\lambda_i$ and assume that the minimum of $F_{i-1}(\lambda_{i-1})+f(\lambda_{i-1},\lambda_i)$ does not lie on $\ell_1$ or $\ell_2$. We want to show that it either lies on or right next to a new breaking point of $F_{i-1}(\lambda_{i-1})$ or on a grid point of $f(\lambda_{i-1},\lambda_i)$. For the sake of contradiction assume otherwise. Then $\lambda_{i-1}$ is in a cell of $f(\lambda_{i-1},\lambda_i)$, that is, by Lemma \ref{lem:f_univariate_structure} we have that $f(\lambda_{i-1},\lambda_i)=\alpha x+c$ for some constant $c$ and a fixed $\alpha\in\{-1,0,1\}$. Further, $\lambda_{i-1}$ is between two new breaking points of $F_{i-1}(\lambda_{i-1})$, thus by Lemma \ref{lem:monotone_between_breakings} we have that $F_{i-1}(\lambda_{i-1})$ is monotonically increasing or decreasing and each piece has integer slope. Thus, $F_{i-1}(\lambda_{i-1})+f(\lambda_{i-1},\lambda_i)$ is monotonically increasing or decreasing except possibly on the first or last part. But then the minimum is either on a new breaking point or right next to it.
\end{proof}

\subsection{The algorithm}
\label{sec:algorithm}
For the algorithm we proceed iteratively over the vertical slabs from left to right.
Ideally, we would want to compute the function $F_i(\lambda_i)$ iteratively for every vertical slab.
The problem though is that by Lemma~\ref{lem:Fstructure} only guarantees us that there are no more than
$O(nm)$ many breakpoints over the whole computation. 
This would pretty directly lead to a $O(nm^2\log m)$ algorithm after showing how to compute the argmin
in which we maintain and update all function values at breakpoints explicitly.
To improve the running time, we next describe a data structure, 
which will allow us to query and update 
the function values $F_{i-1}(\lambda_{i-1})$ for specific breakpoints $\lambda_{i-1}$.
Crucially, while this data structure might contain at some point $O(nm)$ elements,
all its operations can be executed in $O(\log(nm))$ time.

\mypar{Maintaining the breakpoints.} Let $(x, g)$ be a \emph{key-function pair} where $x \in \mathbb R$ and $g : \mathbb R \rightarrow \mathbb R$
is a linear function.
We assume that the linear function $g$ is represented by the real numbers $m$ and $b$ 
where $m$ is the slope and $b$ is the $y$-axis value of the function.
Hence, when we in the following speak of adding two linear functions we
just need to add its slope and $y$-axis value.
To query $F_i$ on the set of up to $O(mn)$ breakpoints we next describe a data structure that
stores $n$ key-function pairs and allows us to efficiently add functions to the pairs and 
evaluate them at a point equivalent to some key stored in the data structure.
Our data structure should support for a pair $(x,g)$ insertion via $\textsc{Insert}(x,g)$,
deleting via $\textsc{Delete}(x,g)$,
evaluating $g$ at $x$ via $\textsc{Evaluate}(x)$,
and the following two operations:
\begin{description}
    \item[{$\textsc{ShiftVertical}(a, b, h)$}] Add for each pair $(x, g)$ with $x \in (a,b)$ the function $h$ to $g$.
    \item[{$\textsc{ShiftHorizontal}(a, b, y)$}] Add for each pair $(x,g)$ with $x \in (a,b)$ the value $y$ to $x$ assuming that the order of elements in the data structure sorted by their key remains the same.
\end{description}

\begin{restatable}{lemma}{datastructure}
    \label{lem:datastructure}
    There is a data structure that maintains a set of $n$ key-function pairs and 
    supports the above operations in $O(\log n)$ time each.
\end{restatable}
\begin{proof}
    At the core of our data structure is a modified AVL tree~\cite{adelson1962algorithm}
    $\mathcal T$ whose keys 
    are going to be the keys of the key-function pairs.
    We assume that each leaf of the AVL tree has a pointer 
    At each internal node $u \in \mathcal T$ we store a \emph{shift-value} $\xi_u \in \mathbb R$ and 
    a linear function $h_u : \mathbb R \rightarrow \mathbb R$ which we call the \emph{function} at this node.

    To show correctness we are going to maintain the following invariant for $\mathcal T$ after each operation:
    \begin{invariant}
    \label{inv:tree}
        Let $(k,g)$ be a key-function pair stored in the data structure then there exists a path starting at the root with vertices $(\xi_1,h_1),\ldots, (\xi_m,h_m)$ in $\mathcal T$ such that
        \begin{enumerate}
            \item $k = \sum_{i=1}^m \xi_i$
            \item $g(k) = \sum_{i=1}^m h_i(k)$
        \end{enumerate}
    \end{invariant}

    We implement $\textsc{Search}(k)$ in our tree such that it returns the node where 
    the key of the left tree is smaller and of the right tree is larger, or the returned node is a leaf.
    We use the following comparison.
    Let $(\xi_i,h_i)$ be the current node, then we recurse on the left subtree when
    $k \leq \xi_i$ and on the right one else using the new key $k - \xi_i$.
    Let $(\xi', g')$ be the returned node and $(\xi_1,h_1),\ldots, (\xi_m,h_m)$ the
    traversed path from the root, then $k = \sum_{i=1}^m \xi_i$ by Invariant~\ref{inv:tree} whenever 
    $(k,g)$ is an element of our data structure.

    Rotation is implemented in the usual way, but we need to update the functions stored at each node.
    Let $u = (\xi,h)$ be the node around which we execute our rotation.
    We first describe a left rotation.
    First, we update the values, then we rotate as usual.
    Let $\xi'$ be the shift-value and $h'$ the function of the right child $u_r$ of $u$.
    We add $\xi'$ to the shift-values of the children of $u_r$ and 
    $h'$ to the function of the children of $u_r$.
    We set the shift-value and function of $u_r$ to be equal to $\xi$ and $h$ respectively and
    finally set shift-value and function to constant $0$ for $u$.
    Right rotation is executed is implemented symmetrically.
    Checking all possible paths to traverse $\mathcal T$ from root to any leaf after rotation
    shows that Invariant~\ref{inv:tree} is fulfilled.
    Using left and right rotations we can implement the necessary double rotations as usual.
 
    Suppose we want to insert the key-function pair $(k,g)$.
    If $\mathcal T$ is empty we create one single node $(k,g)$ and return.
    The invariant then holds by definition.
    Else, we first search for the location to insert $(k,g)$ using the same comparison as in the search function.
    Let $(\xi_1,h_1),\ldots, (\xi_m,h_m)$ be the path starting at the root and ending at $(\xi_m,h_m)$.
    Moreover, let $\xi = \sum_{i=1}^k \xi_i$ and $h = \sum_{i=1}^k h_i$.
    We insert a new node for $(k,g)$ with the values $(k - \xi, g - h)$.
    Invariant~\ref{inv:tree} holds by definition.

    Suppose we want to delete the key-function pair $(k,g)$.
    First, locate $(k,g)$ with $\textsc{Search}(k)$, let $u = (\xi, h)$ be the returned node.
    Again, we only need to describe how to update shift-values and functions,
    the other parts are as usual.
    If $u$ is a leaf we can simply delete it.
    If $u$ has only one child $v$ 
    we simply add the shift-value $\xi$ and function $h$ to the corresponding values of $v$ and replace $u$ by $v$.
    If $u$ has two children, we locate the successor $v$ of $u$ in its right subtree and
    replace $u$ by $v$.
    In case $v$ was also the right child of $u$ we add $\xi$ to $v$'s 
    shift-value and $h$ to $v$'s function.
    It remains to handle the case when $v$ was not the right child of $u$.
    Let $\xi_v$ and $h_v$ be $v's$ shift-value and function.
    Moreover, let $w$ be the unique child of $v$.
    We replace $v$ by $w$ and add $\xi_v$ and $h_v$ to the shift-value and function of $w$.
    Then, we replace $u$ by $v$ and set its shift-value and function equal to $\xi$ and $h$ respectively.

    We argue that Invariant~\ref{inv:tree} holds after executing the above replacements, but
    before any rotations are executed.
    Since the invariant holds after a series of rotations it then holds after the whole call to delete has terminated.
    If $u$ was a leaf Invariant~\ref{inv:tree} holds again.
    Let $(\xi_1,h_1),\ldots, (\xi_m,h_m)$ be a path starting at the root and passing through $u$ before $u$ got deleted.
    The cases of $u$ having had only one child or its successor being a child are easy to check since 
    we only need to add the shift-value and function of $u$ to the node replacing it without any adjustment to the subtrees.
    By the same argument, the invariant holds for the subtree previously rooted at the successor of $u$.
    Finally, in case the successor was not a child of $u$ we simply replaced its shift-value and function by the one of $u$
    which means the invariant holds.

    It remains to describe the two shift functions.
    We begin with $\textsc{ShiftVertical}(a, b, h)$.
    Let $u_a$ and $u_b$ be the two nodes of $\mathcal T$ found by calling $\textsc{Search(a)}$ and $\textsc{Search}(b)$.
    We now add $h$ to the function of the lowest common ancestor $r$ of $u_a$ and $u_b$ in $\mathcal T$.
    In the following we assume that $u_b$ is in the right subtree of $r$ and $u_a$ in the left subtree of $r$.
    If the sum of shift-values on the path from root to $u_a$ is larger than $a$ we
    simply subtract $h$ from the function of the left child of $u_a$,
    else we subtract $h$ from the function of $u_a$ and add it to the function of its right child.
    We proceed symmetrically for $u_b$ with left and right children reversing their roles.

    Let $(k,g)$ be an element of our data structure.
    Part one of Invariant~\ref{inv:tree} holds since the shift-values were not altered.
    Let $u_k = \textsc{Search(k)}$, $u_a = \textsc{Search(a)}$, and $u_b = \textsc{Search(b)}$.
    There are two cases $k \in (a,b)$ or $k \not\in (a,b)$.
    Consider the sums $s_k$, $s_a$, $s_b$ of shift-values along paths from the root to $u_k$, $u_a$, and $u_b$,
    then it holds that $s_a < s_k < s_b$, meaning that $u_k$ is a successor or $u_a$ and a predecessor of $u_b$.
    Hence, the lowest common ancestor $r$ of $u_a$ and $u_b$ is on the root to $u_k$ path,
    which means that $g(k) = \sum_{i=1}^m(h_i(k)) = \sum_{i=1, u_i\neq r}^m h_i(k) + h_r(k)$ where $h_r$ is the function stored at $r$.
    Note that $h_r$ consists of the sum of the function at $r$ 
    before the call to $\textsc{ShiftVertical}$ and $h$ as desired.
    Symmetrically we can prove the statement for $k \not\in (a,b)$.

    Next we describe $\textsc{ShiftHorizontal}(a, b, \xi)$.
    Let $u_a$ and $u_b$ be the two nodes of $\mathcal T$ found by calling $\textsc{Search(a)}$ and $\textsc{Search}(b)$.
    Again, let $r$ be the lowest common ancestor of $u_a$ and $u_b$ in $\mathcal T$.
    We add $\xi$ to the shift-value of $r$ and to its function the constant $q = -(m \cdot \xi)$ where 
    $m$ is the slope of the current function at $r$.
    In the following we assume that $u_b$ is in the right subtree of $r$ and $u_a$ in the left subtree of $r$.
    We update the subtree rooted at $u_a$ as above with respect to the tree before any changes
    subtracting $\xi$ and $q$ from the shift-value and function of the left child of $u_a$ or
    subtracting them from $u_a$ and adding them to the right child of $u_a$.
    We proceed symmetrically for the shift-value and function of the subtree rooted at $u_b$.

    Let $(k,g)$ be an element of our data structure after an application of $\textsc{ShiftHorizontal}$
    Since we assume that the ordering of the key-functions pairs by their keys does not change in this operation
    we can immediately assume that $\mathcal T$ is still fulfills all the properties of an AVL-tree.
    It remains to show that Invariant~\ref{inv:tree} holds.
    Let $u_k = \textsc{Search(k)}$, $u_a = \textsc{Search(a)}$, and $u_b = \textsc{Search(b)}$.
    There are two cases $k \in (a,b)$ or $k\not\in (a,b)$.
    Consider the sums of shift-values $s_k$, $s_a$, and $s_b$ on the respective root-node-paths.
    Before the call to $\textsc{ShiftHorizontal}$ we know that $s_a < s_k < s_b$.
    Hence, $u_k$ is a successor of $u_a$ and a predecessor of $s_b$.
    Moreover, we desire that $s_k = k$ which holds if $s_k = s_k' + \xi$ where $s_k'$ is the sum
    on the root-to-node path to $u_k$ before the call to $\textsc{ShiftHorizontal}$.
    But this is true as either $\xi$ is added once either at $r$ or at $r$ and the right child of $u_a$ in which case its also subtracted at $u_a$.
    In case $k\not\in (a,b)$ we argue in the same fashion, either $u_k$ is not in the subtree rooted at $r$ in which case no change was made to its shift-value or $k$ is in the left subtree of $u_a$ in which case $x$ is first added at $r$ and then subtracted at $u_a$ or at the left child of $u_a$.

    It remains to show part two of Invariant~\ref{inv:tree}.
    After the execution of $\textsc{ShiftHorizontal}$ we desire that $g(k) = g'(k + \xi) = mk + m\xi + b = g'(k) - m\xi $ where
    $g'$ is the function before the call to $\textsc{ShiftHorizontal}$.
    Assume $k \in (a,b)$ and let $g(k) = \sum_{i=1}^m h_i(k)$ on a path from the root to $u_k$, 
    as above consider $h_r$ to be the function at $r$,
    then 
    $\sum_{i=1}^m h_i(k) = \sum_{i=1,u_i\neq r}^m h_i(k) + h_r(k) = \sum_{i=1,u_i\neq r}^m h_i'(k) + h_r'(k) - m\xi = g'(k) - m\xi$.
    In case $k \not\in (a,b)$ the invariant holds we simply recover the previous function as desired.

    Since we only need to ever traverse paths of the tree,
    modify vertices along those paths or their direct children, and 
    do constant time algebraic operations the runtime bounds follow directly from the ones for the AVL-tree.
\end{proof}

\mypar{Putting things together.} From now on we are going to assume that $F_{i-1}$ is given as a data structure as described in above.
Next to inserting and deleting key-function pairs in this structure,
we are allowed to evaluate a function at its key,
add a linear function on an interval of key-function pairs, and 
add constant shifts to the keys of the elements.

\begin{restatable}{lemma}{argminlemma}
\label{lem:argmincomp}
Given $F_{i-1}(\lambda_{i-1})$, we can compute $\Psi(\lambda_i)$ in time $O(m\log m)$.    
\end{restatable}

\begin{proof}
We compute $\Psi(\lambda_i)$ on the vertical boundary lines. There are $O(m)$ such lines by Lemma \ref{lem:f_bivariate_diagram}, so in order to achieve the claimed running time, we have to be able to do this in time $O(\log m)$ per line.
For each of these vertical lines, by Lemma~\ref{lem:Psi_path}, 
$\Psi$ can only be on one of $O(m)$ possible points.
These are all points corresponding to new breaking points of $F_{i-1}$ and
we can compute them in $O(m\log m)$ time by querying $F_{i-1}$ $O(m)$ times.
Moreover, since these are all breaking points from the previous iteration,
we can assume that the corresponding $O(m)$ coordinates are explicitly stored.

We introduce the following four functions:
\begin{itemize}
    \item $\ell^+(x):=x$;
    \item $\ell^-(x):=-x$;
    \item $s^+(x)$ is continuous piecewise linear where all pieces have slope $1$ on covered intervals and $0$ on uncovered intervals;
    \item $s^-(x)$ is continuous piecewise linear where all pieces have slope $-1$ on covered intervals and $0$ on uncovered intervals.
\end{itemize}
By Lemma \ref{lem:f_univariate_structure}, the function $F_{i-1}(\lambda_{i-1})+f(\lambda_{i-1},\lambda_i)$ restricted to any vertical boundary line can be partitioned into four intervals and on each interval it is of the form $F(\lambda_i)+g+c$ for a real value $c$ and a function $g\in\{\ell^+,\ell^-, s^+, s^-\}$. In particular, the argmin restricted to one of these intervals is invariant of the value $c$. Further, the boundaries $p_1,\ldots, p_4$ of the intervals can be read off in constant time.

Given the values from above we can compute all $F_{i-1}(\lambda_{i-1})+g(\lambda_{i-1})$ for $g\in\{\ell^+,\ell^-, s^+, s^-\}$.
We store the values corresponding to one $g$ in a segment tree~\cite{DBLP:books/lib/BergCKO08} that supports querying the minimum on a given interval of leafs in $O(\log m)$ time.
Building each such tree takes time at most $O(m\log m)$.
For a vertical boundary line we now read off the four intervals and for each interval find the minimum within this interval using the segment tree of the corresponding function $F_{i-1}(\lambda_{i-1})+g$ and return the argument.
We then evaluate $F_{i-1}(\lambda_{i-1})+f(\lambda_{i-1},\lambda_i)$ on these four arguments and keep the one minimizing $F_{i-1}(\lambda_{i-1})+f(\lambda_{i-1},\lambda_i)$.
The computed values already uniquely determine $\Psi(\lambda_i)$ except between two vertical boundaries where $\Psi(\lambda_i)$ jumps over a diagonal segment. In these cases we still need to compute the exact location of the jump. This amounts to evaluating two linear functions and computing where they are equal, so it can be done in constant time per diagonal segment, that is, $O(m)$ total time.

\end{proof}

\begin{restatable}{lemma}{fcomputationlemma}
\label{lem:Fcomputation}
Given $F_{i-1}(\lambda_{i-1})$, we can compute $F_i(\lambda_i)$.
\end{restatable}

\begin{proof}
By Lemma~\ref{lem:argmincomp} we can assume that $\Psi(\lambda_i)$ is given and 
since its description consists of only $O(m)$ points we may traverse it in a linear fashion.
We run the following algorithm to perform the necessary updates to $F_{i-1}$.

Consider a vertical segment of $\Psi$ between $\Psi(\lambda_i^{(j)})$ and $\Psi(\lambda_i^{(j+1)})$
we call $\textsc{Delete}(\Psi(\lambda_a))$ on all lines that lie between 
$\Psi(\lambda_i^{(j)})$ and $\Psi(\lambda_i^{(j+1)})$. 
This can be implemented by retrieving all the to-be-deleted nodes of the tree
and then calling delete on all of them.

For each segment between $\lambda_i^{(j)}$ and $\lambda_i^{(j+1)}$ that lies on $\ell_1$ or $\ell_2$
in the plot of $f$ we first execute
$\textsc{ShiftHorizontal}(\Psi(\lambda_i^{(j)}),\Psi(\lambda_i^{(j+1)}), d)$ if we are on $\ell_1$ and  
if we are on $\ell_2$ we execute$\textsc{ShiftHorizontal}(\Psi(\lambda_i^{(j)}),\Psi(\lambda_i^{(j+1)}), -d)$,
where $d$ is the length of the vertical slab between $\Lambda_{i-1}$ and $\Lambda_i$.
Let $g(\lambda_i):=f(\Psi(\lambda_i),\lambda_i)$ and note that by Corollary~\ref{cor:order}, between $\lambda_i^{(j)}$ and $\lambda_i^{(j+1)}$ this is a single linear function of the form $mx+q$ where $m\in\{-1,0,1\}$.
Next, we call $\textsc{ShiftVertical}(\lambda_i^{(j)}, \lambda_i^{(j+1)}, g(\lambda_i))$.

Finally, call $\textsc{Insert}(\lambda_i^{(j)}, w_j)$ for every $\lambda_i^{(j)}$ where $w_j$ is the
explicit value for $F_i(\lambda_i^{(j+1)})$.

To show correctness we are going to argue that for every key $k$ in our data structure
it holds $F_i(k) = \textsc{Evaluate}(k)$.
Whenever $k$ was inserted in the last step of the algorithm this follows by correctness of the insert function.
For all other elements we only executed one horizontal and one vertical shift at most.
It remains to show that these shift values are indeed correct.
Which follows with
\[F_i(\lambda_i)=F_{i-1}(\lambda_i \pm d) + f(\Psi(\lambda_i), \lambda_i) = F_{i-1}(\lambda_{i-1}) + f(\Psi(\lambda_i), \lambda_i).\qedhere\]
\end{proof}

\begin{theorem}
Given two polygonal curves $P$ and $Q$ with $n$ and $m$ vertices, respectively, 
the semi-discrete Barking distance of $P$ to $Q$ can be computed in time $O(nm\log(nm))$.
\end{theorem}

\begin{proof}
The semi-discrete Barking distance of $P$ to $Q$ is just the value $F_n(|Q|)$. We can compute $F_n(\lambda_n)$ by iteratively computing all $F_i(\lambda_i)$. Examining the proof of Lemma~\ref{lem:Fcomputation} 
we can see that the needed computations can be implemented in $O(m\log m)$.
Additionally, every query to the data structure takes only $O(\log(nm))$ time by Lemma~\ref{lem:datastructure}.
Finally, by Lemma~\ref{lem:Fstructure} there are only $O(nm)$ breaking points and 
each one of them is inserted and deleted at most once from our data structure.
It follows that the total running time is in $O(nm\log(nm))$.
\end{proof}

\section{The continuous setting}\label{sec:continuous}

The underlying ideas are similar to those for the semi-discrete setting.
The definitions of the functions $f(\lambda_{i-1},\lambda_i)$ and $F_i(\lambda_i):=\min_{\lambda_{i-1}}(F_{i-1}(\lambda_{i-1})+f(\lambda_{i-1},\lambda_i))$ extend.
In contrast to before, these functions are now piecewise quadratic.
The BSD has a natural grid structure, where each rectangular cell corresponds to a choice of a segment $P_i$ of $P$ and a choice of a segment $Q_i$ of $Q$. It is a well-known fact (see e.g.\ Lemma 3 in \cite{alt1995computing}) that for each such rectangle $R_{i,j}$, the clear space of the BSD is the intersection of $R_{i,j}$ with an ellipse. We denote this clear space by $C_{i,j}$ and call it \emph{cut ellipse}. In the following, we restrict our attention to a fixed segment $P_i$ of $P$, so we abbreviate the notation to $C_j$.

The main idea for the algorithm is that we can always find an optimal curve in the BSD such that each segment corresponds to either (i) a segment of maximum or minimum slope or (ii) a segment on a special diameter of a $C_j$. %
We discretize the problem by describing all the potential endpoints of such segments. 
For each cut ellipse $C_j$ we consider the \emph{longest feasible diameter} $d(C_j)$ given by a segment whose supporting line passes through the center of the (uncut) ellipse with endpoints coinciding with the intersections of the two vertical tangents of $C_j$.
Denote by $l_j$ and $r_j$ the left and right endpoints of this diameter, respectively.

We again denote by $\Lambda_{i-1}$ and $\Lambda_i$ the left and right boundary of the slab. For any curve $\gamma$ in the BSD we define its cost as the time $\gamma$ spends outside $\bigcup C_j$. We say that a curve with start point $a$ and endpoint $b$ is \emph{optimal} if it minimizes the cost among all such paths.

\begin{restatable}{lemma}{optimalpath}
    \label{lem:optpath}
    For any two points $\lambda_{i-1}$ and $\lambda_i$ on the boundary of a vertical slab 
    there is an optimal polygonal curve from $\lambda_{i-1}$ to $\lambda_i$
    such that its segments are either of maximum or minimum slope (Type 1) or
    lies on $d(C_j)$ for some cut ellipse $C_j$ (Type 2).
\end{restatable}
\begin{proof}
    Let $\gamma$ be a polygonal path from $\lambda_{i-1}$ to $\lambda_i$ that 
    does not confirm to some of the three conditions.

    Subdivide $\gamma$ with points $\gamma_1,\ldots,\gamma_k$ whenever $\gamma$ intersects 
    the boundary of a cut ellipse.
    Let $\gamma_j$ and $\gamma_{j+1}$ be two consecutive points such that
    the segments on the sub-path $\gamma_{j,j+1}$ are not all of Type~1 and~2.
    
    First, we consider the case that $\gamma_{j,j+1}$ lies in an obstacle.
    Then, we can replace $\gamma_{j,j+1}$ by any polygonal path with segments of Type 1 and 2 that connects the points $\gamma_j$ and $\gamma_{j+1}$. If the new path ever leaves the obstacle, the the original path $\gamma$ was not optimal.

    The other case is that $\gamma_{j,j+1}$ lies in cut ellipse $C$.
    First, assume that $d(C)$ has absolute slope smaller than the maximal absolute slope.
    We distinguish two cases, either $\gamma_{j,j+1}$ intersects $d(C)$ or it doesn't.
    In the first case, we replace $\gamma_{j,j+1}$ by 
    a segment $s_j$ of maximal absolute slope with endpoint at $\gamma_j$, 
    a segment $s_{j+1}$ of maximal absolute slope with endpoint at $\gamma_{j+1}$, and
    a segment $s'$ along $d(C)$ connecting the intersection points of $d(C)$ with $s_j$ and $s_{j+1}$.
    In the latter case, $\gamma_{j,j+1}$ is replaced by the same path or
    by the two segments $s_j$ and $s_{j+1}$ solely.
    
    Second, assume that $d(C)$ has absolute slope larger than the maximal absolute slope.
    If the segment with endpoints $\gamma_j$ and $\gamma_{j+1}$ has maximal absolute slope we 
    replace $\gamma_{j,j+1}$ by this segment.
    Assume this is not the case, then
    we replace $\gamma_{j,j+1}$ by three segments of maximum absolute slope.
    A segment $s_j$ of maximal absolute slope with endpoint at $\gamma_j$, 
    a segment $s_{j+1}$ of maximal absolute slope with endpoint at $\gamma_{j+1}$, and
    a segment $s'$ connecting the endpoints of $s_j$ and $s_{j+1}$ at maximum absolute slope.
    Note that $s_j$, $s'$, and $s_{j+1}$ can
    always be chosen in this way since the connection of $s_j$ and $s_{j+1}$ along $d(C)$ is too steep,
    but the segment connecting $\gamma_j$ and $\gamma_{j+1}$ is shallower than the maximum absolute slope.
    Hence, there exists a segment in between lying inside $C$ and having maximum absolute slope.
\end{proof}

We now further restrict the optimal curves.
For each ellipse $C_j$ completely contained in the vertical slab we consider eight lines in total, four with maximum slope and four with minimum slope.
They are the two tangents $t_l^+$ and $t_r^+$ of maximum and 
the two tangents $t_l^-$ and $t_r^-$ of minimum slope,
with the intersection point of $t_l^{+/-}$ before the one of $t_r^{+/-}$ in $x$-direction.
The other four lines are the lines of maximum and minimum slope passing through $l_j$ and $r_j$,
denoted as $s_l^+$, $s_r^+$, $s_l^-$, and $s_r^-$.
For a cut ellipse $C_j$ that is cut off, we add the same lines, even if $l_j$ or $r_j$ are not in the vertical slab.
We additionally add lines of maximum and minimum slope passing through all intersections of $C_j$ with grid lines of the BSD.
Let $S^+$ and $S^{-}$ be the resulting set of lines with positive and negative slope, respectively.
For a cut ellipse $C_j$, we call all these up to 24 lines its \emph{boundary lines}.

$S^+$ ($S^{-}$) defines a set of parallel \emph{strips} that cover the slab.
We define the following lines only for $S^+$, the definition for $S^{-}$ is symmetric.
By definition, in each of them any maximum slope line intersects the same subset of cut ellipses. 
Consider a strip $\tau$ bounded from the left by line $t_l\in S^+$ and from the right by $t_r\in S^+$ and two cut ellipses $C_i,C_j$ in $\tau$, we denote by $c_{\tau,C_i,C_j}(x)$ the cost function of going between the boundaries of $C_i$ and $C_j$ using a maximum slope line at horizontal distance $x$ from $t_l$. 
This function is the sum of the costs of traversing the obstacles between consecutive cut ellipse boundaries that lie in between $C_i$ and $C_j$. 
It is thus a sum of quadratic functions, and therefore a quadratic function itself.
We denote by $\opt_{i,j,\tau}^+$ the line of maximal slope which minimizes the cost $c_{\tau,C_i,C_j}(x)$. Note that if $\opt_{i,j,\tau}^+$ is not equal to $t_l$ or $t_r$, then it is unique.
We call $\opt_{i,j,\tau}^{+}$ the \emph{positive transition line} and 
$\opt_{i,j,\tau}^{-}$ the \emph{negative transition line} of $C_i$ and $C_j$ in $\tau$.

\begin{restatable}{lemma}{shiftlemma}
\label{lem:shift}
For any two points $\lambda_{i-1}$ and $\lambda_i$ on the boundary of a vertical slab 
there is an optimal polygonal curve as in Lemma~\ref{lem:optpath} such that all segments of Type~1 in obstacles lie either on lines of the type $\opt_{i,j,\tau}^{+/-}$, a boundary line, 
or on the four lines of maximum and minimum slope at the start and end points of the polygonal curve.
\end{restatable}
\begin{proof}
Let $\gamma$ be an optimal polygonal path from $\lambda_{i-1}$ to $\lambda_i$ as in Lemma~\ref{lem:optpath}
for which at least one of its segments of Type~1 lies in an obstacle and 
neither on a transition line nor on one of the lines of maximal and minimal slope at the start or end 
vertices of $\gamma$.
Let $s$ be the first segment seen from the start of $\gamma$ that lies in an obstacle, but 
does not lie on one of our desired lines.
We assume without loss of generality that $s$ has maximal slope, the case where $s$ has minimal slope is symmetric.
If one endpoint of $s$ coincides with one of the endpoints of $\gamma$, then $s$ lies on the line of maximal slope through this endpoint.
Let $\tau$ be the strip with bounding lines $\tau_l$ and $\tau_r$ in which $s$ lies.
We try to shift $s$ towards $\tau_l$.
Note that cost will not improve as we assumed $\gamma$ to be optimal.
If we can shift until the segment before $s$ degenerates to a single point, we have reduced the number of vertices on the polygonal path $\gamma$ and we continue inductively.
If we cannot shift further because the cost would increase, there are two cases.
If the cost increases in both directions of shifting, then $s$ lies on the positive transition line between two cut ellipses.
If the cost increases only towards $\tau_l$, then $s$ lies on a boundary line.%
\end{proof}

\begin{observation}
    \label{obs:linebounds}
    There are $O(m)$ boundary lines each of which can be computed in time $O(1)$.
    There are $O(m^2)$ transition lines and these can be computed as minima of $O(m)$ quadratic functions in time $O(m^3)$.
\end{observation}

We are now ready to define our discretization of the problem. On a high level, we define a weighted directed graph $G_i$ connecting relevant points on $\Lambda_{i-1}$ with relevant points on $\Lambda_i$ such that optimal curves correspond to shortest paths in $G_i$.
The edges of $G_i$ are geometrically motivated by the boundary and transition lines defined above. In order to also include the information of $F_{i-1}(\lambda_{i-1})$, we further do a construction similar to the boundary lines, but for the non-differentiable points of $F_{i-1}$.
Formally, let $x\in\Lambda_{i-1}$ be a non-differentiable point of $F_{i-1}(\lambda_{i-1})$. We consider two lines through $x$, one with maximal slope and one with minimal slope. Taking these two lines for every non-differentiable point of $F_{i-1}(\lambda_{i-1})$, we get a set of lines that we call the \emph{$F$-induced lines.}
We define the graph $G_i$ (and a natural straight-line drawing of it in the BSD).
First, we define the following set of vertices and their embedding in the BSD.
\begin{itemize}
    \item[(i)] For each intersection between two lines which are boundary or $F$-induced lines lying outside or on the boundary of $\bigcup_i C_i$, we add a vertex.
    \item[(ii)] For each intersection of a boundary line or an $F$-induced line with $\Lambda_{i-1}$ or $\Lambda_i$, we add a vertex. This includes adding a vertex on $\Lambda_{i-1}$ for each non-differentiable point of $F_{i-1}$.
    \item[(iii)] For each $C_j$ we add as vertices the intersection points of the boundary of $C_j$ with all boundary and $F$-induced lines.
    \item[(iv)] For each $C_j$ add as vertices the intersections of its boundary with its transition lines.
    \item[(v)] For each interval between two vertices on $\Lambda_i$ we add one vertex in the interval on $\Lambda_i$. For each of those vertices, we now add the lines of maximal and minimal slope to the set of $F$-induced lines and add vertices as in (i), (ii) and (iii).
    \item[(vi)] Finally, we add one dummy vertex $L$, which we think of as lying to the left of $\Lambda_{i-1}$.
\end{itemize}
It will follow from our arguments that the function $F_i(\lambda_i)$ has at most $O(im)$ pieces.
Thus we get $O(n^2m^2)$ vertices in total in the vertical slab.
We add four types of edges: 
\begin{itemize}
    \item We add edges between each vertex of Types (i), (ii), (iii), and (v) and the at most four closest vertices on the lines of maximum or minimum slope that pass through this vertex. We only add such edges if they fully lie in obstacles.
    \item For each transition line between cut ellipses $C_i$ and $C_j$ we add an edge between the two consecutive vertices of Type (iv) on the line where one belongs to $C_i$ and the other to~$C_j$.
    \item Inside each $C_i$ we make all possible connections, i.e., we connect two vertices lying on the boundary of $C_i$ by an edge if the slope of the segment connecting them does not exceed our slope bounds. 
    Note that every $C_i$ is convex and that we might be including edges in a $C_i$ incident to vertices that lie in the shared boundary between $C_i$ and $C_{i-1}$ or $C_{i+1}$.
    \item We connect $L$ with all vertices lying on $\Lambda_{i-1}$. 
\end{itemize}
Note that this defines $O(n^2m^2)$ edges.
We orient all these edges from left to right
and assign them weights.
Consider an edge $e$ not incident to $L$ and its intersections with cut ellipses. 
These intersections define sub-segments, some of which fully lie in obstacles. 
The weight of $e$, $w(e)$, is the sum of the lengths of the projections onto the horizontal line of these sub-segments in obstacles.
For the edges incident to $L$ we define the weight of the edge $\{L,v\}$ as $F_{i-1}(v)$.

\begin{restatable}{lemma}{nondifflemma}
\label{lem:left-nondiff}
For any fixed $\lambda_i$ corresponding to a point on $\Lambda_i$ and for each interval $I=[\lambda_{i-1},\lambda_{i-1}']$ on $\Lambda_{i-1}$ that minimizes $f(\lambda_{i-1},\lambda_i)$ 
we have that 
the endpoints of $I$ are vertices of $G_i$.
\end{restatable}
\begin{proof}
Let $\lambda_{i-1}$ be a point which minimizes $f(\lambda_{i-1},\lambda_i)$ but that is not a vertex of $G_i$. Let $\gamma$ be a trajectory of optimal cost. 
By Lemma~\ref{lem:shift} we may assume that the first segment is of maximal or minimal slope; without loss of generality we assume that it is of maximal slope. 
We first consider the case in which $\lambda_{i-1}$ lies on the intersection of $\Lambda_i$ with some cut ellipse $C_j$. The intersection of $\Lambda_i$ with $C_j$ is in general a segment, and its endpoints are vertices of $G_i$. 
Thus, if $\lambda_{i-1}$ is not a vertex of $G_i$, we have that in a small neighborhood of $\lambda_{i-1}$ all points lie on the boundary of $C_j$, implying that $\lambda_{i-1}$ is in the interior of a minimizing interval whose endpoints are vertices of $G_i$.

Next we assume that $\lambda_{i-1}$ lies in some obstacle. Let $C_j$ be the cut ellipse first intersected by the first segment of $\gamma$ and let $p$ be the intersection point. Consider the tangent $\tau$ to $C_j$ at $p$ and assume it to be non-vertical having positive slope; the argument for negative slope is symmetric.
Since $\tau$ is of positive slope the leftmost point of $C_j$ lies to the left of $p$. Hence, there exists a number $\varepsilon > 0$ such that we can shift the first segment of $\gamma$ by $\varepsilon$ units to the bottom and elongate the second segment of $\gamma$ with a segment lying completely inside $C_j$. Let $p'$ be the intersection point of this shifted segment and $C_j$, then $p'$ lies left of $p$ and we found a path of cost smaller than the one of $\gamma$, a contradiction.

Finally, assume that $\tau$ is vertical at $p$. Then, $p$ coincides with the left-most point on the boundary of $C_j$. But then, we find the starting point of $\gamma$ to lie on the intersection of a boundary-line with $\lambda_{i-1}$ for which a vertex was added to $G_i$.
\end{proof}

\begin{restatable}{lemma}{graphcontlemma}
\label{lem:graph-continuous}
Let $v$ be a vertex of $G_i$ on $\Lambda_i$ and let $w(L,v)$ be the weight of a shortest path in $G_i$ from $L$ to $v$. Then $w(L,v)=F_i(v)$.
\end{restatable}
\begin{proof}
Let $\pi$ be a path from $L$ to $v$ in $G_i$ whose weight is $w(L,v)$. By our definition of $G_i$, each edge $e_k$ of the path defines a sub-trajectory $\gamma_k$ whose cost equals the weight of $e_k$. These sub-trajectories intersect in the vertices of $\pi$, so $\pi$ defines a trajectory $\gamma$ to $v$ with cost $w(L,v)$. This shows that $w(L,v)\geq F_i(v)$.

Let now $\gamma$ be an optimal trajectory to $v$, which by definition has cost $F_i(v)$. 
Recall that $F_i(v) = \min_{w}\{F_{i-1}(w) + f(w,v)\}$ and that the non-differentiable points of $F_{i-1}$ are vertices of $G_i$.
Thus, by Lemma~\ref{lem:left-nondiff}, we can assume that the point $w$ in which $\gamma$ intersects $\Lambda_{i-1}$ is a vertex of $G_i$. 
By Lemma~\ref{lem:shift}, we can further assume that $\gamma$ fulfills the conditions of this Lemma and thus, by definition of $G_i$, $\gamma$ is partitioned into sub-trajectories that are in one-to-one correspondence with edges in $G_i$. 
More precisely, we partition $\gamma$ at its intersections with a cut ellipse $C_j$ as long as the sub-trajectory does not follow a transition line.
These edges define a path from $L$ to $v$ in $G_i$ whose weight is $F_i(v)$, showing that $w(L,v)\leq F_i(v)$.
\end{proof}

With this we can already compute $F_i(\lambda_i)$ for any $\lambda_i$ that is a vertex of $G_i$. In the following, we discuss how we can extend the computation to all of $F_i(\lambda_i)$.
Recall that $F_i(\lambda_i)$ is a continuous, piecewise quadratic function. 
In particular, in order to compute $F_i(\lambda_i)$, we will also need to compute the non-differentiable points of it.
We first observe that $G_i$ already gives us many of these.
For this, let $V_{\Lambda_i}$ denote the vertices of $G_i$ on $\Lambda_i$ that are not of Type (v).
Note that there is a natural order on $V_{\Lambda_i}$ given by the order on $\Lambda_i$.

In order to compute all the non-differentiable points of $F_i(\lambda_i)$ we realize that in the interval $(v,w)$ of $\Lambda_i$ between consecutive vertices of $V_{\Lambda_i}$, the optimal curves come in two forms: either the last segment is of maximal or of minimal slope. We define the functions $F_i^+(\lambda_{i})$ and $F_i^{-}(\lambda_i)$ as the cost of an optimal curve to $\lambda_i$ where we restrict the set of optimal curves to those of Lemma~\ref{lem:shift} with the last segment having maximal slope and minimal slope, respectively.
It then follows from Lemma~\ref{lem:shift} that $F_i(\lambda_i)=\min(F_i^+(\lambda_i),F_i^{-}(\lambda_i))$.
The next lemma tells us that if there is a non-differentiable point of $F_i(\lambda_i)$ in the interior of $[v,w]$ it must be at the point in which $F_i^+(\lambda_{i})=F_i^{-}(\lambda_i)$.
We further define the graph $G_i^+$ ($G_i^-$) as the subgraph of $G_i$ in which we restrict the edges incident to vertices on $\Lambda_i$ to those that have positive slope (resp. negative slope) and observe that we can compute $F_i^+(\lambda_{i})$ and $F_i^{-}(\lambda_i)$ from $G_i^+$ and $G_i^-$ just like we computed $F_i(\lambda_{i})$ from $G_i$.

\begin{restatable}{lemma}{fplusminuslemma}
\label{lem:F-plus-minus}
Let $v$ and $w$ be two consecutive points in $V_{\Lambda_i}$. Then $F_i^+(\lambda_i)$ and $F_i^{-}(\lambda_i)$, restricted to the interval on $\Lambda_i$ where they are not infinite, are continuous, piecewise quadratic functions that do not have any non-differentiable points in the interior of the interval $[v,w]$.
Further, for any vertex $v$ in $G_i$ on $\Lambda_i$ let $w(L,v)$ be the weight of a shortest path in $G_i^+$ ($G_i^{-}$) from $L$ to $v$. Then $w(L,v)=F_i^+(v)$ ($w(L,v)=F_i^{-}(v)$).
\end{restatable}
\begin{proof}
We only argue about $F_i^+(\lambda_i)$, the arguments for $F_i^{-}(\lambda_i)$ are symmetric.
We first argue that $F_i^+(\lambda_i)$ is continuous and piecewise quadratic. As by definition we have that $F_i(\lambda_i):=\min_{\lambda_{i-1}}(F_{i-1}(\lambda_{i-1})+f(\lambda_{i-1},\lambda_i))$, it is enough to show that $f(\lambda_{i-1},\lambda_i)$ is continuous and piecewise quadratic.
This follows from the definition of $f$ as the sum of distances within obstacles of an optimal trajectory, and the fact that the boundaries of the obstacles are piecewise quadratic.
The fact that all non-differentiable points are on vertices of $G_i$ is analogous to the proof of Lemma \ref{lem:left-nondiff}.
Finally, the proof that $w(L,v)=F_i^+(v)$ is analogous to the proof of Lemma \ref{lem:graph-continuous}.
\end{proof}

We now 
describe our algorithm.
For each interval $[v,w]$ between two consecutive points in $V_{\Lambda_i}$ 
we have three values of $F_i^+(\lambda_i)$ and $F_i^{-}(\lambda_i)$: at $v$, at $w$, and at the subdivision vertex of Type (v). 
By Lemma~\ref{lem:F-plus-minus}, since the functions restricted to $[v,w]$ are quadratic functions, we can obtain them via interpolation
and we can compute both $F_i^+(\lambda_i)$ and $F_i^{-}(\lambda_i)$ via a shortest path computation in the graphs $G_i^+$ and $G_i^{-}$.
As both of them have $O(n^2m^2)$ many vertices edges, each of the $O(nm)$ values of $F_i^+(\lambda_i)$ and $F_i^{-}(\lambda_i)$ can be computed in time $O(n^2m^2\log (nm))$ time using Dijkstra's algorithm.
Thus, both $F_i^+(\lambda_i)$ and $F_i^{-}(\lambda_i)$ can be computed in time $O(n^3m^3\log (nm))$.
We further get that both $F_i^+(\lambda_i)$ and $F_i^{-}(\lambda_i)$ have $O(nm)$ many pieces.
As stated before, $F_i(\lambda_i)=\min(F_i^+(\lambda_i),F_i^{-}(\lambda_i))$, thus we can compute $F_i(\lambda_i)$ from $F_i^+(\lambda_i)$ and $F_i^{-}(\lambda_i)$ in additional time $O(nm)$.

\begin{lemma}\label{lem:Fcomputation_continuous}
Given $F_{i-1}(\lambda_{i-1})$, we can compute $F_i(\lambda_i)$ in time $O(n^3m^3\log (nm))$.
\end{lemma}

As before we get an algorithm to compute $F(\lambda_n)$, albeit with higher time complexity.

\begin{theorem}
Given two polygonal curves $P$ and $Q$ with $n$ and $m$ vertices, respectively, the continuous Barking distance of $Q$ to $P$ can be computed in time $O(n^4m^3\log (nm))$.
\end{theorem}

\section{Lower bound}\label{sec:lower_bounds}

Many lower bounds on time complexity are based on the Strong Exponential Time Hypothesis (SETH). 
SETH asserts that for every $\varepsilon>0$ there is a $k$ such that $k$-SAT cannot be solved in $O((2-\varepsilon)^n$ time.
Given two sets $A,B \subseteq \{0,1\}^d$ of $n$ vectors each, the Orthogonal Vectors (OV) problem asks whether there are two orthogonal vectors $u \in A$ and $v \in B$, that is, $u\cdot v = \sum_{i=1}^d u_iv_i = 0$. 
Assuming $d=\omega(\log n)$, the Orthogonal Vectors Conjecture (OVC) states that for all $\epsilon>0$, solving OV requires $\Omega(n^{2-\epsilon})$ time. 
SETH implies both OVC~\cite{Williams04CS}, meaning that the Orthogonal Vectors Conjecture is weaker and thus more likely to be true.
\begin{restatable}{theorem}{lowerboundthm}
Let $P$ and $Q$ be two disjoint polygonal curves with $n$ vertices. 
Assuming OVC,
solving the barking decision problem where the maximum speed of the dog matches the speed of the human and with constant barking radius $\rho$ requires $\Omega(n^{2-\epsilon})$ time for any $\epsilon>0$.    
\end{restatable} 
\begin{proof}
Consider an instance of the OV problem consisting of sets $A,B \subseteq \{0,1\}^d$ of $n$ vectors and $d=\omega(\log n)$. 
We construct an instance of the barking decision problem consisting of two polylines $P$ and $Q$ with $O(n)$ vertices and rational coordinates. %
The trajectories for the human and the dog go along these polylines $P$ and $Q$, respectively.

\begin{figure}
    \centering
    \includegraphics[page=1,width=0.95\textwidth]{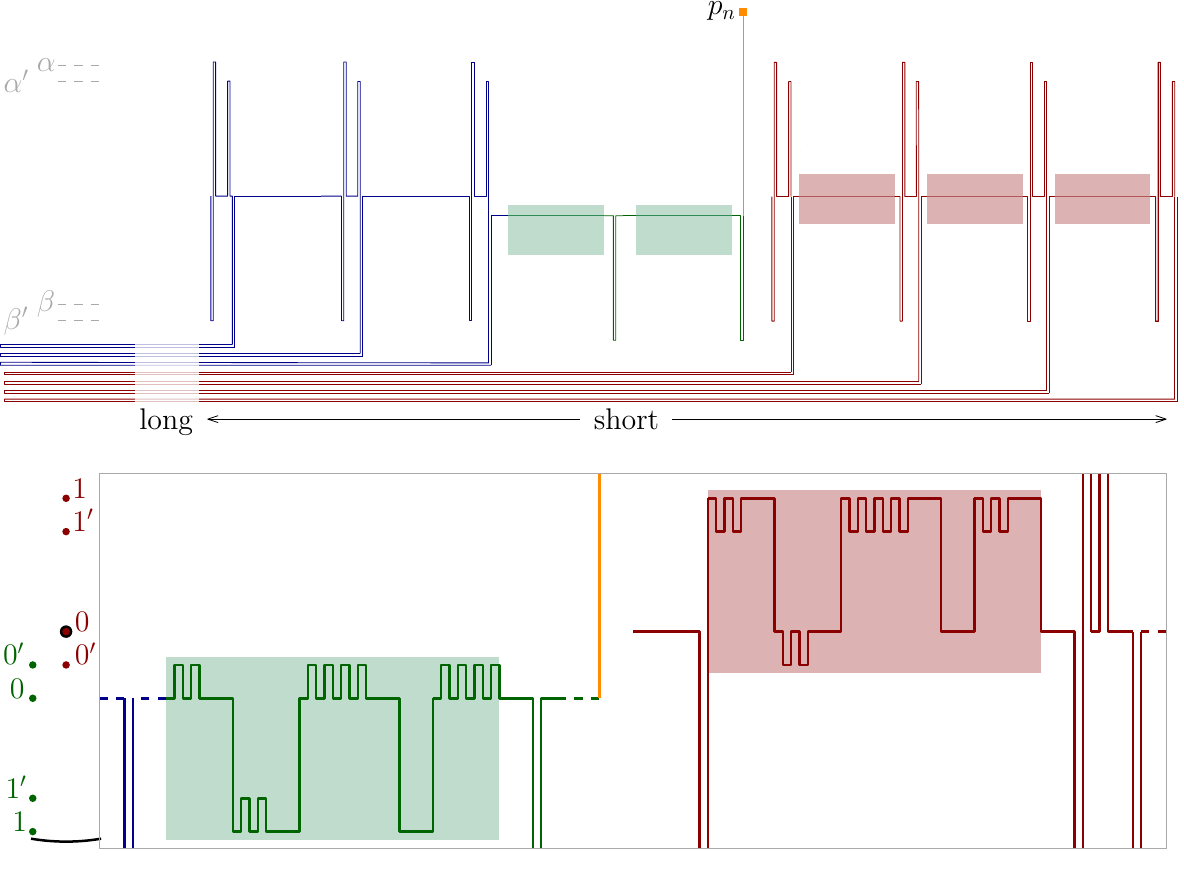}
    \caption{Illustration of the reduction. The bottom gadgets show the polylines corresponding to $00110001000 \in A$ (left) and $01001110100 \in B$ (right).}
    \label{fig:reduction}
\end{figure}

\mypar{Vector gadgets.}
As it is usually the case, the reduction is based on gadgets corresponding to the vectors; refer to Figure~\ref{fig:reduction}. 
To build them, we define a distance $u=\rho/25$.
The polyline $P$ has a gadget for each vector in $A$ and analogously the polyline $Q$ has a gadget for each vector in $B$. 
In these gadgets the horizontal distances are small and thus irrelevant.
We start defining a horizontal line $\kappa_{0'}=\pi_{0'}$ and, symmetrically on the top and bottom, in that order,
lines $\kappa_0$ and $\pi_0$ at distance $4u$,
lines $\kappa_{1'}$ and $\pi_{1'}$ at distance $16u$, and
lines $\kappa_{1}$ and $\pi_{1}$ at distance $20u$. 
\begin{observation}
\label{obs:vectors}
The vertical distance between any two horizontal lines such that at least one of them has subindex $0$ or $0'$ is at most $24u = 24/25 \rho < \rho$. 
In contrast, the distance between two lines with subindices $1$ or $1'$ is at least $32u = 32\rho/25 > \rho$. 
\end{observation}

For vector $v\in A$ (in $B$) we create a subpolyline that for each value $i \in \{0,1\}$ in $v$ has a constant horizontal segment on $\pi_i$ (resp. $\kappa_i$) if the next value is the same, 
and otherwise has an alternation $\pi_i,\pi_{i'},\pi_i,\pi_{i'},\pi_i$ (resp. $\kappa_i,\kappa_{i'},\kappa_i,\kappa_{i'},\kappa_i$) of horizontal segments connected vertically. 
The parts corresponding to different consecutive values are connected with vertical segments. 
Figure~\ref{fig:reduction} (bottom) shows an example.

\mypar{Prefixes and suffixes.}
To force synchronization in the transversal of the vector gadgets, we prepend the vector $1\ldots10$ with $d$ $1$s at the beginning and we append the reverse vector $01\ldots1$ at the end of each vector in $B$. Analogously, we prepend the vector $0\ldots01$ with $d$ $0$s at the beginning and we append the reverse vector $10\ldots0$ at the end of each vector in $A$.
Only in these prefixes and suffixes of $A$ we have such a long stretch (of length at least $4ud$) in which $A$ is within radius $\rho$ of vertices on $\kappa_1$ and $\kappa_1'$.
\begin{observation}
\label{obs:prefix-sufix}
    If always staying within barking radius $\rho$, the dog can only (fully) traverse a prefix or suffix gadget while the human is traversing one. 
    Moreover, two prefix (two suffix) gadgets can only be traversed synchronously. More precisely, when the human is at the last (resp. first) vertex of its gadget in $A$, the dog must be at the part of the gadget corresponding to the digit $0$. 
\end{observation}

\mypar{Split gadgets.}
Between two vector gadgets in $P$, the polyline briefly goes vertically down a distance $\ell$ until a horizontal line $\beta'$. The polyline then goes vertically up until $\pi_0$.
To prevent interference with the vector gadgets, this distance $\ell$ between $\pi_0$ and $\beta'$ should be large enough; $\ell=2\rho= 50u$ suffices. This is the $P$-split gadget. 

Between two vector gadgets in $Q$, the polyline does four detours from $\kappa_0$. 
\begin{enumerate}[(i)]
    \item First, it briefly goes vertically down a distance $\ell$ until a horizontal line $\beta$ and returns vertically to $\kappa_0$. 
    \item Then, polyline then goes vertically up a distance $\ell+10u$ until a horizontal line $\alpha$ and returns vertically. 
    \item After that, the polyline goes vertically up again, this time a distance $\ell$ until reaching the horizontal line $\alpha'$ and then it returns vertically to $\kappa_0$.
    \item The polyline continues down vertically surpassing $\beta'$. 
It then goes horizontally left a distance $d$; it suffices to also make $d=2\rho=50u$. 
From there the polyline traces its way back as it came, first horizontally and then vertically. 
\end{enumerate}
This is the $Q$-split gadget.
This gadget is also included at the start of $Q$, and also at the end (there, including just the two first detours suffices). 

\mypar{Goal.}
At the end of $P$ the polyline has a $P$-split gadget and then it goes vertically up a distance $\ell + 30u$ until point $p_n$. 
\begin{observation}
\label{obs:goal}
The vertical distance from $p_n$ to vertices in $Q$ is only smaller than $\rho$ for vertices on $\alpha$.
\end{observation}

\mypar{Positioning gadget.}
The beginning of polyline $P$ consists of the polyline that would correspond to a vector $B$ which was a set of $m-1$ zero vectors with a full $Q$-split gadget at the end. 
The human traversing this positioning gadget allows the dog to place itself at any $Q$-split gadget (except the last half one) and wait there until the human finishes the positioning gadget. 

\begin{claim}
\label{clm:positioning}
There is a trajectory for the dog such that when the human finishes the positioning gadget, the dog stands at the end of the $Q$-split gadget before a given vector gadget in $Q$ and such that the distance between human and dog is less than $\rho$. 
\end{claim}

At the beginning, the dog moves forward at the same speed as the human, thus staying within barking range. 
When in this way the dog reaches the desired $Q$-split gadget, 
instead of traversing Detour (iv), it returns back to point $q$ between Detours (iii) and (iv), unless the desired $Q$-split gadget is the one before the last vector gadget, in which case the dog continues forward. 
After this, when the human is at $\kappa_0$, so is the dog at point $q$ between Detours (iii) and (iv) of the $Q$-split gadget; when the human goes down to $\beta$, the dog does the same using Detour (iv); 
when the human goes up to $\alpha$ so does the dog using Detour (iii), returning to $q$; 
and when the human goes down and to the left, the dogs does the same, returning again to $q$. 
Only when the human does its last detour of the positioning gadget, the dog fully traverses Detour (iv). \qed

With the construction description completed, notice that the polylines $P$ and $Q$ do not intersect.

\mypar{Correctness.}
We show that the answer to the OV problem is \textit{yes} if and only if the answer to the barking decision problem is \textit{yes}. 
Assume that there are two orthogonal vectors $a,b$ in $A$ and $B$, respectively. 
Note that the vectors are still orthogonal after adding the prefixes and suffixes. 
Claim \ref{clm:positioning} guarantees that the dog can stand at the beginning of the gadget corresponding to $a$ when the human reaches the beginning of the gadget corresponding to $b$. Then, human and dog traverse those gadgets simultaneously. 
When the human traverses the next $P$-split gadget, the dog fully traverses Detour (i) of the next $Q$-split gadget. 
Then the dog stays within barking radius at the point between Detours (i) and (ii) or going down along Detour (i) and returning to the same point while the human traverses the remaining gadgets. 
Finally, when the human goes up to $p_n$ the dog follows Detour (ii) until its top point, thus staying always within barking radius $\rho$.

In the other direction, assume for contradiction that there are no orthogonal vectors in $A$ and $B$ 
and the dog always stays within barking radius $\rho$. 
When the human finishes traversing the positioning gadget, which ends with its Detour (iv), 
the dog must just have traversed Detour (iv) of a $Q$-split gadget $Q_b$, since the (leftmost points of) Detours (iv) (both in $P$ and $Q$) are only within barking radius of themselves (and their close vicinity). 
Moreover, after the human has finished the positioning gadget, 
the dog can never again fully traverse Detours (iii) or (iv) of a $Q$-split gadget. 
Thus, after that, the dog is constrained between the two $Q$-split gadgets $Q_b$ and $Q_{b+1}$. 
Furthermore, it can only reach Detours (iii) and (iv) of $Q_b$, Detours (i) and (ii) of $Q_{b+1}$, and the subcurve in-between. 
Whenever the human traverses a $P$-split gadget, 
the dog must be at a $Q$-split gadget, more precisely, $Q_b$ or $Q_{b+1}$. 
Observations~\ref{obs:vectors} and~\ref{obs:prefix-sufix} imply that the dog can only traverse from $Q_b$ or $Q_{b+1}$ if $B$ has an orthogonal vector in $A$. 
Thus, it cannot reach Detour (ii) of $Q_{b+1}$. 
However, Observation~\ref{obs:goal} implies that only points in the Detour (ii) of a $Q$-split gadget are within barking distance of $q_n$. Thus, we reached a contradiction.

We remark that in the reduction as presented, the dog does not necessarily end at its final position. If we would desire that, we can just add a gadget to the human curve at the end analogous to the positioning gadget. 
As before, $p_n$ can be reached by a trajectory that stays within barking radius if and only if there are two orthogonal vectors in $A$ and $B$. 
The appended positioning gadget in $P$ (same as the one described minus the two first detours) allows the dog to reach the end of $Q$ analogously to the proof of Claim~\ref{clm:positioning}. 
\end{proof}

\section{Conclusion}
We introduced the barking distance to detect outliers in curves and have shown algorithms to compute this measure as well as a conditional lower bound for the computation time. In the discrete and semi-discrete setting, the runtime of our algorithms match the lower bound up to logarithmic factors. For the continuous setting we give an algorithm that is likely not optimal. We believe that using techniques as in Section~\ref{sec:semidiscrete} we can improve the runtime to $O(nm^3\log m)$, but this would still leave a significant gap 
between upper and lower bound.
While we conjecture that it is possible to obtain an $O(nm\log(nm))$ algorithm, it is likely that new ideas are necessary for this.
It would also be interesting to find more efficient algorithms in the continuous setting for restricted types of curves such as time series.
Throughout our paper, we assumed that the barking radius $\rho$ and the speed bound $s$ are fixed. Considering them as variables leads to other interesting algorithmic problems 
where we ask for the minimal speed or barking radius required for the dog such that the human can hear it the entire time. We leave the study of these problems for future work.%

\bibliography{references}

\end{document}